\documentclass[lettersize,journal]{IEEEtran}
\usepackage{amsmath,amsfonts,amssymb,bm}
\usepackage{algorithm}
\usepackage{algorithmic}
\usepackage{array}
\usepackage[caption=false,font=small,labelfont=rm,textfont=rm]{subfig}
\usepackage{textcomp}
\usepackage{stfloats}
\usepackage{xcolor}
\usepackage{url}
\usepackage{dsfont}
\usepackage{todonotes}
\usepackage{verbatim}
\usepackage{graphicx}
\usepackage{cite}
\usepackage{booktabs}
\usepackage{multirow}
\usepackage{stmaryrd}

\usepackage{tikz}
\usetikzlibrary{positioning,arrows.meta,fit,calc,shapes.geometric}
\usetikzlibrary{decorations.pathreplacing}

\usepackage{theorem}
{\theorembodyfont{\rmfamily}
   }
{\theorembodyfont{\rmfamily}
   \newtheorem{The}{{\textbf Theorem}}}
{\theorembodyfont{\rmfamily}
   }
{\theorembodyfont{\rmfamily}
   \newtheorem{Lem}{{\textbf Lemma}}}
{\theorembodyfont{\rmfamily}
   \newtheorem{Cor}{{\textbf Corollary}}}
{\theorembodyfont{\upshape}
   }
{\theorembodyfont{\upshape}
   }
{\theorembodyfont{\rmfamily}
   \newtheorem{Ass}{{\textbf Remark}}}

\usepackage{pgfplots}

\usepackage[acronym,shortcuts]{glossaries}
\pgfplotsset{compat=1.18}

\newacronym{AWGN}{AWGN}{additive white Gaussian noise}
\newacronym{BP}{BP}{belief propagation}
\newacronym{BPSK}{BPSK}{binary phase-shift keying}
\newacronym{CNN}{CNN}{convolutional neural network}
\newacronym{CRC}{CRC}{cyclic redundancy check}
\newacronym{C2C}{C2C}{control-to-control}
\newacronym{DL}{DL}{deep learning}
\newacronym{GLRT}{GLRT}{generalized likelihood ratio test}
\newacronym{LDPC}{LDPC}{low-density parity-check}
\newacronym{LLR}{LLR}{log-likelihood ratio}
\newacronym{LoS}{LoS}{line-of-sight}
\newacronym{LRT}{LRT}{likelihood ratio test}
\newacronym{MIMO}{MIMO}{multiple-input multiple-output}
\newacronym{mMTC}{mMTC}{massive machine-type communications}
\newacronym{MRC}{MRC}{maximal ratio combining}
\newacronym{NN}{NN}{neural network}
\newacronym{NMS}{NMS}{normalized min-sum}
\newacronym{PGF}{PGF}{probability generating function}
\newacronym{PLR}{PLR}{packet loss rate}
\newacronym{ReLU}{ReLU}{rectified linear unit}
\newacronym{ROC}{ROC}{receiver operating characteristic}
\newacronym{SF}{SF}{superframe}
\newacronym{SIC}{SIC}{successive interference cancellation}
\newacronym{SISO}{SISO}{soft-input soft-output}
\newacronym{SVM}{SVM}{support vector machine}
\newacronym{URLLC}{URLLC}{ultra-reliable and low-latency communication}
\newacronym{VF}{VF}{virtual frame}
\newacronym{FLOP}{FLOP}{floating point operation}

\def\BibTeX{{\rm B\kern-.05em{\sc i\kern-.025em b}\kern-.08em
    T\kern-.1667em\lower.7ex\hbox{E}\kern-.125emX}}
\usepackage{balance}

\begin{document}
\title{A Deep Learning-based Receiver for Asynchronous  Grant-Free Random Access in Control-to-Control Networks}
\author{Massimo Battaglioni, Edoardo Carnevali, Dania De Crescenzo, Enrico Testi, Marco Baldi, Enrico Paolini
\thanks{An earlier version of this paper was presented in part at the IEEE Conference on Standards for Communications and Networking (CSCN), September 2025 [DOI: 10.1109/CSCN67557.2025.11230747].\\
The work of Massimo Battaglioni, Dania De Crescenzo, Enrico Testi, and Enrico Paolini was supported in part by the European Union under the Italian National Recovery and Resilience Plan (NRRP) of NextGenerationEU, partnership on ``Telecommunications of the Future'' (PE00000001 - program ``RESTART'').\\
Massimo Battaglioni, Edoardo Carnevali, and Marco Baldi are with Dipartimento di Ingegneria dell'Informazione, Universit\`a Politecnica delle Marche, 60131 Ancona, Italy. Dania De Crescenzo and Enrico Paolini are with CNIT/WiLab, DEI, University of Bologna,
47522 Cesena, Italy. Enrico Testi is with FixIA s.r.l.. Corresponding author: Massimo Battaglioni (email: m.battaglioni@univpm.it).}}

\maketitle

\begin{abstract}
In this paper, we study grant-free, asynchronous control-to-control (C2C) communications in an indoor scenario with a shared wireless channel. Each communication node transmits command units, each consisting of a variable-length low-density parity-check (LDPC)--coded payload preceded by a start sequence and followed by a tail sequence. Due to the asynchronous nature of the access, transmissions from different nodes are not aligned over time. As a result, each receiving controller observes the superposition of multiple command units transmitted by different nodes over a receiver-defined superframe interval. Each node transmits one or more replicas of the same command unit. We propose a receiver architecture in which the detection of command unit boundaries (start/tail sequences) is carried out by a single convolutional neural network (CNN) operating directly on the received signal. We show that, while start-sequence detection must rely only on the received waveform, tail-sequence detection can additionally exploit the soft information produced by the LDPC decoder, together with channel estimates. Finally, once commands units are  successfully decoded, successive interference cancellation (SIC) can be applied. Simulation results demonstrate that the receiver we propose achieves reliable packet-boundary identification and a low end-to-end packet loss rate, even under uncoordinated and high-traffic operating conditions.\\ 
\end{abstract}

\begin{IEEEkeywords}
Control-to-control communications, deep learning, grant-free random access.
\end{IEEEkeywords}

\section{Introduction}

Modern motion control systems are increasingly designed as collections of autonomous yet interdependent subsystems, where coordination emerges from local interactions rather than from global supervision \cite{Aceto2019}. This paradigm shift is particularly evident in autonomous robotics, cyber-physical platforms, and industrial automation, where stringent real-time and reliability constraints render centralized control impractical \cite{Cuozzo2025}. 
Also modular manufacturing systems and large-scale industrial machinery often operate under conditions which make centralized coordination either infeasible or undesirable. Controllers may be added or removed during operation (hot-plugging), communication interactions can vary over time, and traffic characteristics range from periodic control updates to sporadic event-driven messages.

\Ac{C2C} communications allow motion controllers to exchange state information and coordination data directly, without relying on centralized schedulers or strictly time-driven control loops. While this approach improves adaptability, fault tolerance, and responsiveness to topology changes, it fundamentally reshapes the design space of the communication layer, introducing new challenges related to scalability, interference, and computational overhead. As a result, the design of \ac{C2C} communication architectures has a direct impact on the ability of distributed motion subsystems to operate safely and meet real-time constraints, especially in dynamic and decentralized environments \cite{3gppTR22804, 5GACIA_IIoT_2021}.

In this context, grant-free asynchronous access over the wireless channel reduces protocol complexity and shortens reaction times, but shifts a significant portion of the system burden to the physical layer. In the absence of centralized scheduling and global time references, the receiver must cope with partially overlapping transmissions, uncertain arrival times, besides channel variations over time. As a result, reliable message detection and decoding must be achieved under increased interference and timing uncertainty, making physical-layer design a critical aspect of such communication systems. This paper proposes a solution to these physical-layer challenges for grant-free, asynchronous \ac{C2C} communications.

\subsection{Related works}

One scenario that has been extensively studied in the existing literature is that of \ac{mMTC}, in which multiple access techniques aim to address the fundamental problem of scalability, which makes uncoordinated, grant-free approaches preferable to solutions based on static resource allocation, even at the cost of tolerating high error rates and long latencies \cite{Choi2022,Liva2024}.
On the other hand, \ac{URLLC} applications are based on coordinated access and are primarily designed to achieve low latency, but this does not allow them to easily scale as the number of devices grows.

The solution we explore in this paper represents a trade-off between these two approaches.
In fact, we consider a scenario in which there are numerous nodes (though not massive in number) making orchestration inefficient, and where there are requirements for latency and reliability.
This requires the use of short codes to achieve low latency, yet capable of handling packets of variable length, unlike what occurs in coordinated systems.
Furthermore, the proposed approach is fully decentralized, and thus capable of eliminating any bottlenecks and single points of failure resulting from the presence of a single central controller.

Asynchronous random access and grant-free transmission schemes have been extensively investigated in the literature. Early works such as \cite{DeGaud2014} extend contention resolution diversity slotted ALOHA to the asynchronous domain by exploiting packet replication and \ac{SIC}. Similarly, grant-free access frameworks, e.g., \cite{Azari2017}, exploit time and frequency asynchronism together with replica diversity to enable massive connectivity, relying on correlation-based detection and \ac{SIC} at the receiver. From an abstract point of view, graph-based methods (e.g., \cite{Paolini2015}) model random access as iterative decoding on sparse graphs. However, they typically operate at the packet level, treating collisions as erasures or resolving them via replica-based \ac{SIC}, without addressing waveform-level detection and decoding. 

In parallel, some contributions have addressed communication design for industrial and control-oriented scenarios. For instance, \cite{Chang2021} focuses on the joint design of communication and control by mapping control requirements into reliability constraints and optimizing device activation probabilities, while MAC-layer solutions such as \cite{Gao2021} develop scheduling and access protocols to support dense deployments and stringent quality-of-service requirements. These approaches operate at a higher abstraction level and do not consider the interaction between asynchronous access, signal detection, and channel decoding.
Along similar lines, \cite{Almonacid2017} addresses asynchronism within a protocol-oriented design based on packet replicas and interference resolution, implicitly assuming that packet detection and alignment are available. The work \cite{Azari2021} studies asynchronous grant-free access through a protocol-level and stochastic-geometry framework with replica-based transmissions, without considering signal-level processing techniques to recover overlapping packets. In \cite{Hu2024}, the authors handle asynchronism via MAC-layer retransmission strategies coordinated by a central access point, treating collisions as  undecodable events. Other works operate closer to the signal level, yet differ substantially from the approach proposed here. In \cite{Zhu2021}, asynchronous access is treated as a sparse estimation problem aimed at identifying active users and their delays via a learned approximate message-passing architecture; however, recovering entire coded transmissions goes beyond reveiver operations at the signal level, which is the focus of the present work. In \cite{Jiang2022}, interference is handled via spatial processing at a centralized massive-\ac{MIMO} base station in a slotted uplink, relying on spatial separation and coordination rather than waveform structure.   Regarding the use of channel coding, \cite{Ebrahimi2017} employs \ac{LDPC} codes within a frame-synchronous random access scheme with known transmission patterns at the protocol level, which is fundamentally different from the signal-level role that channel coding plays in the proposed receiver.

In contrast to prior work, this paper tackles asynchronous grant-free access at the signal level, modeling symbol-level interference and proposing a joint chain for detection, channel estimation, soft decoding, and \ac{SIC}. The approach is fully decentralized and leverages waveform and code structure, with tight coupling between detection and decoding to refine packet boundaries without idealized packet-level assumptions. Deep learning-based methods for grant-free random access have been proposed in \cite{deSou2023,Khan2024,Decre2025}. However, these techniques address isolated tasks of the access procedure, such as activity detection, preamble detection, or replica detection/combining, each requiring a dedicated \ac{NN}. 
By contrast, our approach employs a single \ac{NN} that jointly performs start-sequence and tail-sequence detection within a unified end-to-end framework, reducing architectural complexity while capturing the structural dependencies between the two tasks. Moreover, decoder-side information is explicitly used for tail-sequence detection, introducing a stronger interaction between synchronization and channel decoding than in previous approaches.

\subsection{Our contribution}

We consider an unsourced random access communication scenario \cite{Yuri2017}, in which transmissions are not associated with explicit transmitter identities and are decoded solely based on their content. Controllers access the channel asynchronously and without prior coordination, all employing a common codebook; consequently, the receiver cannot rely on protocol-layer framing, scheduling, or metadata to delineate individual messages. To enable physical-layer message segmentation under these conditions, and inspired by telecommand synchronization and coding schemes used in deep-space communications \cite{bluebook}, we assume that each controller transmits a self-contained command unit composed of three elements: a start sequence for detection and synchronization, a \emph{variable-length}  \ac{LDPC}-coded payload, and a known tail sequence marking the end of the transmission.  Building on this design, we propose a receiver architecture that performs transmission-boundary recovery directly at the physical layer, thereby avoiding centralized control and preserving the low-latency and scalability properties required by grant-free \ac{C2C} communication. The setting we consider, therefore, differs significantly from traditional designs, where metadata or scheduling ensures message alignment. In the context we consider, all structure is deduced from the signal itself.

To address this problem, we propose a receiver architecture in which a single \ac{CNN} is employed to detect both the start and the tail sequence of each command unit at different stages of reception. Unlike the above existing approaches, which are limited to fixed-length transmissions, the receiver we propose is explicitly designed to handle variable-length coded payloads, for which end-of-packet detection at the physical layer constitutes a central challenge. The \ac{CNN} operates in coordination with the \ac{LDPC} decoder by exploiting soft information derived from the decoding process, including \acp{LLR}. In particular, decoder feedback is used to validate the detection of the tail sequence and to reliably determine transmission termination in the presence of asynchronous access and interference.
Numerical simulations show that the approach we propose for start and tail sequence detection can outperform traditional methods, such as correlators and \ac{LRT}-based detectors \cite{Chiani2006}, as well as recently proposed solutions, such as decoder-based  \cite{giuliani2025access} and machine-learning-based \cite{battaglioni2025mltail} detectors.

The main distinction with respect to the conference version of this paper \cite{battaglioni_cscn} lies in the scope of the considered receiver functionality. While \cite{battaglioni_cscn} primarily addresses boundary detection, the present work encompasses the complete receiver chain, from signal detection to channel decoding. This broader perspective enables performance assessment at the system level, using communication-level reliability metrics rather than detection accuracy alone. Moreover, it allows the identification of critical operating conditions that most adversely affect packet recovery, as well as the theoretical  characterization of the probability of occurrence of such events.

\subsection{Notation}\label{subsec:not}

The set of complex and real numbers are denoted by $\mathbb{C}$ and $\mathbb{R}$, respectively. The notation $[a,b]$ represents the set of integers between $a$ and $b$, inclusive. All logarithms are assumed to be base-$e$, unless otherwise specified. For a logical condition $\mathcal{P}$, the Iverson bracket $\llbracket \mathcal{P} \rrbracket$ equals $1$ if $\mathcal{P}$ is true, and $0$ otherwise.

We denote vectors by bold lowercase letters (e.g., $\mathbf{x}$), and matrices by bold uppercase letters (e.g., $\mathbf{H}$). The transpose and the Hermitian (conjugate transpose) operators are denoted by $(\cdot)^\top$ and $(\cdot)^\mathrm{H}$, respectively.  Given a vector $\mathbf{x} \in \mathbb{C}^N$, we use $x[i]$ to denote its $i$-th element, with $i \in [0,N-1]$ unless otherwise specified. We denote the length-$R$ all-one (resp., all-zero) column vector by $\mathbf{1}_R$ (resp., $\mathbf{0}_R$). 

 Random variables are indicated using standard notation (e.g., $w[i] \sim \mathcal{CN}(0,N_0)$ denotes a circularly symmetric complex Gaussian variable with zero mean and variance $N_0$). Expectation is denoted by $\mathbb{E}[\cdot]$.

Throughout the paper, time is treated as discrete, and all signals are indexed accordingly. Transmissions from controller $j$ are denoted as $\mathbf{x}^{(j)}$, and asynchronous delays are represented by integer offsets $\tau_j \in \mathbb{Z}$. The received signal at a designated receiver is denoted by $\mathbf{r}$.

\subsection{Paper outline}

The paper is organized as follows.  The system model is introduced in Section \ref{sec:sys_mod}. A theoretical analysis of the occurrence probability of the most detrimental operating conditions for the receiver is presented in Section~\ref{sec:the}.   The \ac{CNN}-based architecture for start and tail sequence detection is presented in Section \ref{sec:bound}.  Section~\ref{sec:numres} reports a comprehensive set of results of numerical simulations, while some concluding remarks are provided in Section \ref{sec:concl}.

\section{System model}\label{sec:sys_mod}

We consider a wireless system deployed in a three-dimensional (3D) indoor environment, where an unspecified number of controllers operate over a shared communication medium. Due to hot-plugging capabilities, controllers may enter or leave the environment at any time, resulting in a dynamic and a priori unknown number of active devices (see Fig. \ref{fig:C2Cscenario} for an example). Within this setting, we model an asynchronous, grant-free \ac{C2C} communication scenario representative of such distributed motion systems. Communication occurs over short-range wireless links characterized by flat Rician fading. Each controller is equipped with $R$ antennas, with a single antenna used for transmission and all $R$ antennas simultaneously employed by the receiver to exploit spatial diversity.  

\begin{figure}[t]
    \centering
    \includegraphics[width=0.7\linewidth]{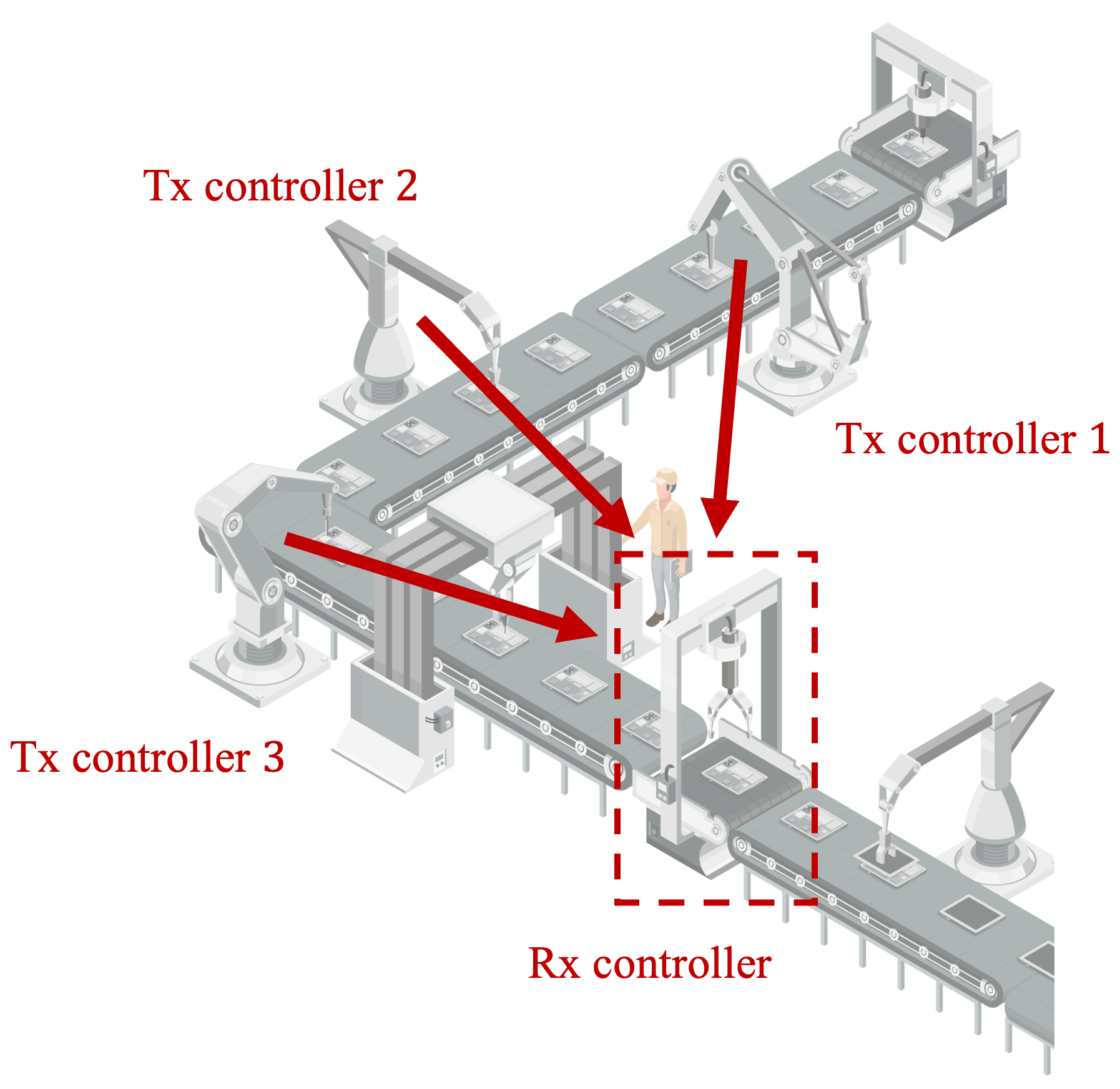}
    \caption{Example of a \ac{C2C} communication setting in which multiple upstream supervisory controllers convey task-specific directives to a downstream execution controller responsible for real-time coordination of assembly-line processes.}
    \label{fig:C2Cscenario}
\end{figure}

In the architecture we consider, motion controllers exchange \ac{LDPC}-coded messages without either coordination or centralized scheduling. Each transmission consists of a self-contained communication unit comprising a start sequence, a variable number of sequentially concatenated \ac{LDPC} codewords, and a tail sequence that marks the end of the transmission. Use of the tail sequence enables reliable detection of the end of the coded data block without relying on upper-layer protocols.

\subsection{Transmission model}

The controllers are assumed to be randomly positioned within a bounded 3D parallelepiped of dimensions \( L_x \), \( L_y \), and \( L_z \). The spatial separation between controllers is characterized by the inter-controller distance \( d_i \), which arises from uniformly distributed node locations within this volume. Accordingly, \( d_i \) is modeled as a uniformly distributed random variable,
\[
d_i \sim \mathcal{U}(0, d_{\max}),
\]
where \( d_{\max} = \sqrt{L_x^2 + L_y^2 + L_z^2} \) denotes the maximum possible separation, corresponding to the length of the space diagonal of the parallelepiped.

\begin{figure}[t]
    \centering
    \includegraphics[width=0.9\linewidth]{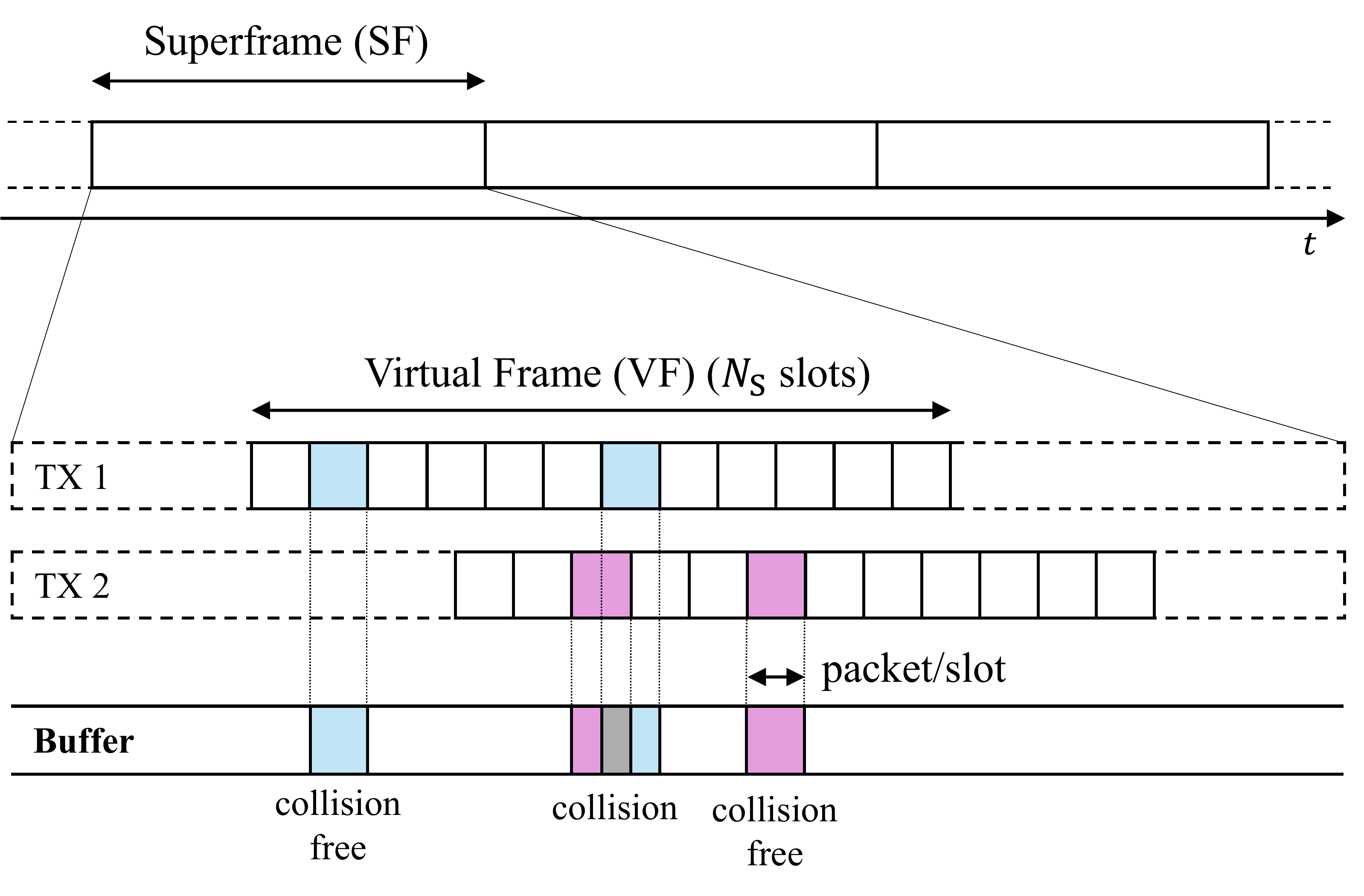}
    \caption{Representation of the hierarchical time organization into SFs, VFs, and slots, for a given receiving controller.}
    \label{fig:time}
\end{figure}

As illustrated in Fig.~\ref{fig:time}, time is hierarchically structured into \acp{SF}, \acp{VF}, and transmission slots. Each \ac{SF} comprises multiple \acp{VF}, and each \ac{VF} is further partitioned into a finite number of transmission slots used for message replication. The duration of a \ac{SF}, measured in symbol intervals, is not fixed a priori and is denoted by \( T_{\mathrm{SF}} \). Each \ac{VF} has a duration
\[
T_{\mathrm{VF}} = N_{\mathrm{S}} T_{\mathrm{S}},
\]
where \( T_{\mathrm{S}} \) is the slot duration in symbol intervals and \( N_{\mathrm{S}} \) is the number of slots per \ac{VF}. By construction, the hierarchy satisfies \( T_{\mathrm{SF}} > T_{\mathrm{VF}} \).

\subsubsection{Superframes}

Each \ac{SF} is analyzed from the perspective of a single reference controller operating in reception mode for the entire duration of the \ac{SF}. We assume that the signals intended for this controller can be distinguished from those associated with other controllers. This assumption is purely analytical and does not imply any explicit receiver selection, coordination, or scheduling among controllers. It allows us to characterize the reception process at an individual controller while abstracting away the simultaneous operation of other receivers in the network.

This modeling choice differs from the approach adopted in \cite{battaglioni_cscn}, where each \ac{SF} has a fixed duration and a single controller is randomly selected at its beginning and designated as the only receiver for the entire \ac{SF}.

Multiple access from the other controllers is modeled through a contention-based transmission process. Within a \ac{SF} associated with a given receiver, each controller in standby mode
decides independently whether to initiate a transmission at the beginning of each symbol interval. Specifically, a controller becomes active with probability \( p \), according to a memoryless Bernoulli process. This leads to an average activation rate of \( p \) activations per symbol interval and converges to a Poisson process in the limit of vanishing symbol durations. 

Upon activation, a controller is assigned a \ac{VF} of fixed duration and remains active for its entire length. During this period, it transmits \( N_{\mathrm{rep}} \) replicas of the same message.  Once the assigned \ac{VF} has elapsed and all scheduled transmissions are completed, the controller returns to standby mode and can be later reactivated in the same way.

\subsubsection{Virtual frames}

Motivated by the random-access frameworks in \cite{Casini2007, Khan2024}, we adopt a replication-based transmission strategy to enhance reliability under asynchronous and uncoordinated access. When a controller becomes active, it transmits \( N_{\mathrm{rep}} \) replicas of the same message within a \ac{VF}, which consists of \( N_{\mathrm{S}} \) transmission slots. Each active controller independently and uniformly selects \( N_{\mathrm{rep}} \) distinct slots from the set \([0,N_{\mathrm{S}}-1]\), placing one replica in each selected slot using its single transmit antenna. Slot selection is performed independently across controllers, and \acp{VF} are initiated asynchronously, without coordination or explicit signaling toward the receiver. The latter, equipped with \( N_R \) antennas, attempts to decode any replica that is successfully received.

Under the assumption
of independent replica placement and finite replica duration,  the superposition
of active transmissions at the receiver results in a random number of replicas
that overlap at any given time instant. 

\subsubsection{Slots}

Each slot within a \ac{VF} accommodates up to \( L_{\max} \) symbols, corresponding to the maximum packet length determined a priori by the applicative scenario. If the actual packet length is \( L < L_{\max} \), the transmitting controller randomly selects the packet starting position within the slot, uniformly over the interval \( [0,\, L_{\max} - L] \).

\subsubsection{Command unit structure}

The variable-length commands, representing the information bits, are first encoded into codewords using block \ac{LDPC} coding. Each command unit contains $\iota \cdot k$ information bits, where $\iota$ is a positive integer denoting the number of codewords contained in a single command unit. These codewords are then modulated using \ac{BPSK}, yielding a transmitted sequence
\[
\mathbf{x}_{\text{code}} \in \{-1,1\}^{\iota \cdot L_{\text{code}}},
\]
where \( L_{\text{code}} \) denotes the LDPC code block length. We assume that $L_{\text{code}}$ divides $L_{\max}$, so that the maximum number of encoded codewords in the command unit is $\iota_{\max}\triangleq\frac{L_{\max}}{L_{\text{code}}}$.

In total, each command unit transmitted by a controller is composed of $L = L_{\text{pre}} + \iota \cdot L_{\text{code}} + L_{\text{tail}}$ symbols, where:
\begin{itemize}
    \item $L_{\text{pre}}$ is the length of the start sequence;
    \item $L_{\text{code}}$ is the length of an \ac{LDPC} codeword;
    \item $L_{\text{tail}}$ is the length of the tail sequence.
\end{itemize}

Under the above assumptions, the transmitted signal corresponding to a single command unit can be expressed as
\[
\mathbf{x} = [\mathbf{x}_{\text{pre}}, \mathbf{x}_{\text{code}}, \mathbf{x}_{\text{tail}}] \in \{-1,1\}^L.
\]

\subsubsection{LDPC Coding}

The $\iota \cdot k$ information bits of a command are encoded into $\iota$ codewords, each of length $L_{\mathrm{code}}$, using a binary $(L_{\mathrm{code}},k)$ \ac{LDPC} code. Decoding is performed using \ac{BP}. In particular, the \ac{NMS} \cite{FossoNMS} soft-input iterative decoding algorithm running (exactly\footnote{This choice is motivated by the use of the decoder output in the tail-sequence detector. To ensure consistency, the decoder must be evaluated at a fixed iteration. Early stopping would produce outputs at different convergence stages, resulting in non-comparable decoder states across candidates.}) $i_{\max}$ iterations, which can also yield intermediate states useful for the tail sequence detection, is employed.

\subsection{Channel model}

Consider the $i$-th transmitter. Let $\mathbf{h}_i \in \mathbb{C}^{R}$ 
denote the flat-fading channel vector between the transmitter and the $R$ receive antennas. The channel is assumed to be \emph{block-static} over a transmission block of length $L$ symbols, i.e.,
\begin{equation}
\mathbf{h}_i[t] = \mathbf{h}_i, \quad \forall\, t \in [0, L-1].
\end{equation}
This assumption reflects the structured and relatively stable arrangement of motion subsystems in industrial environments, where channel variations are negligible over the duration of a single transmission block. We further assume that multiple replicas transmitted by the same controller experience identical large-scale fading, including common path loss, deterministic \ac{LoS} components, and average scattering power. In contrast, small-scale fading varies independently across replicas. Specifically, for the $j$-th replica, the channel vector is modeled as
\begin{equation}
\mathbf{h}_i^{(j)} = \mu_i \mathbf{1}_R + \mathbf{g}_i^{(j)},
\end{equation}
where $\mu_i \in \mathbb{C}$ represents the deterministic \ac{LoS} component associated with transmitter $i$, and
\begin{equation}
\mathbf{g}_i^{(j)} \sim \mathcal{CN}\!\left(\mathbf{0}, \sigma_i^2 \mathbf{I}_R\right)
\end{equation}
models the replica-dependent small-scale fading. Equivalently, each antenna coefficient satisfies
\begin{equation}
h_{i,r}^{(j)} \sim \mathcal{CN}(\mu_i, \sigma_i^2), \quad r \in [0, R-1],
\end{equation}
with independent realizations across replicas $j \in [0, N_{\mathrm{rep}}-1]$. The corresponding Rician $K$-factor for transmitter $i$ is defined as
\begin{equation}
K_i = \frac{|\mu_i|^2}{2\sigma_i^2}.
\end{equation}
The total average received power from transmitter $i$ is:
\begin{equation}
P_i = 2\sigma_i^2 (K_i + 1),
\end{equation}
and accounts for large-scale path loss and log-normal shadowing according to
\begin{equation}
P_i = P_0 \, \gamma_i \left( \frac{d_{\mathrm{ref}}}{d_i} \right)^{\beta},
\qquad
P_0 = P_t \left( \frac{\lambda_c}{4\pi d_{\mathrm{ref}}} \right)^2,
\label{eq:logno}
\end{equation}
where $P_t$ is the transmit power, $\lambda_c$ is the carrier wavelength, $d_i$ is the distance between transmitter $i$ and the receiver, $\beta$ is the path loss exponent, and $d_{\mathrm{ref}}$ is a reference distance. The shadowing term $\gamma_i$ is modeled as a log-normal random variable:
\begin{equation}
10 \log_{10} \gamma_i \sim \mathcal{N}(0, \sigma_{\mathrm{dB}}^2),
\end{equation}
with $\sigma_{\mathrm{dB}}$ denoting the shadowing standard deviation in the dB domain. We also consider \ac{AWGN}, modeled as
\begin{equation}
\mathbf{w} \sim \mathcal{CN}\!\left(\mathbf{0}, \sigma_w^2 \mathbf{I}_R\right),
\end{equation} where
$\sigma_w^2$ denotes the noise variance per receive antenna and $\mathbf{I}_R$ is the identity matrix with side $R$.

\subsection{Receiver operations}

The receiver performs a sequence of steps to detect, estimate, and decode incoming command frames under asynchronous, uncoordinated transmissions.  

The receiving process begins with start sequence detection, whose goal is to identify the beginning of a command frame. The estimated starting time is denoted by \( \hat{\tau} \in \mathbb{Z} \). Detection can be implemented using traditional correlation-based methods, such as \ac{GLRT}, or through data-driven techniques, such as properly trained deep \acp{NN}. The performance of these approaches is evaluated experimentally in Section~\ref{sec:numres}.

Upon detection of a start sequence, the
receiver extracts the corresponding segment of received samples:
\begin{equation}
\mathbf{R}_{\mathrm{pre}}
=
\big[
\mathbf{r}[\hat{\tau}],
\mathbf{r}[\hat{\tau}+1],
\ldots,
\mathbf{r}[\hat{\tau}+L_{\mathrm{pre}}-1]
\big]
\in \mathbb{C}^{R \times L_{\mathrm{pre}}}.
\end{equation}
 Using the known start 
sequence $\mathbf{x}_{\mathrm{pre}} \in \{-1,1\}^{L_{\mathrm{pre}}}$, the channel
vector is estimated according to a least-squares criterion (which also coincides with the maximum-likelihood estimator under \ac{AWGN}) as
\begin{equation}
\hat{\mathbf{h}}
=
\frac{\mathbf{R}_{\mathrm{pre}}\,\mathbf{x}_{\mathrm{pre}}^{\mathrm{H}}}
{\|\mathbf{x}_{\mathrm{pre}}\|^2}.
\label{eq:ch_estimate}
\end{equation}
The estimated channel vector $\hat{\mathbf{h}}$ is then employed to coherently
combine the signals received across the $R$ antennas. In particular, for each
symbol time $t \geq \hat{\tau}$, the transmitted symbol is estimated via
\ac{MRC} as
\begin{equation}
\hat{x}[t-\hat{\tau}]
=
\frac{\hat{\mathbf{h}}^{\mathrm{H}}\,\mathbf{r}[t]}
{\|\hat{\mathbf{h}}\|^2}.
\label{eq:mrc_estimation}
\end{equation}

Then, the symbols following $\mathbf{R}_{\mathrm{pre}}$ are segmented into $L_{\text{code}}$-bit words, which are then passed to the \ac{LDPC} decoder. To detect the end of a command frame, the receiver uses a dedicated tail detector operating on each $L_{\text{code}}$-bit block. This component attempts to match the observed signal to the known tail sequence to identify the frame termination point. As in the case of start sequence detection, both classic \ac{GLRT}-based methods and learning-based approaches can be used for tail detection, depending on the desired trade-off between complexity and accuracy.

After completing a traversal over the entire \ac{SF}, the receiver  reconstructs the corresponding transmitted signals
and subtracts them from the received waveform, enabling \ac{SIC}. This interference cancellation step produces a cleaned version
of the received \ac{SF}. The detector and the channel decoder are then applied again to the cleaned
signal, and the procedure is iterated until no additional start-tail sequence
pairs can be detected. This iterative \ac{SIC} process progressively reduces multi-user
interference and significantly improves the probability of successful decoding in
high-load asynchronous scenarios.

The receiver architecture is summarized in Fig. \ref{fig:receiver_scheme}. Information between the \ac{LDPC} decoder and the tail detector is exchanged in the learning-based approach (dashed arrow), while, when \ac{GLRT}-based detection is employed, these two components operate independently.

\begin{figure*}[t]
  \centering
  \resizebox{\linewidth}{!}{%

\begin{tikzpicture}[
  font=\large,
  >=Latex,
  node distance=18mm and 18mm,
  block/.style={draw, rounded corners=4mm, thick, fill=blue!70,
                minimum height=18mm, minimum width=38mm,
                align=center, text=white},
  decision/.style={draw, diamond, thick, fill=blue!70,
                   aspect=2, align=center, text=white,
                   inner sep=1pt},
  terminator/.style={circle, draw, thick, fill=gray!70,
                     minimum size=6mm},
  line/.style={-Latex, thick, draw=gray!70},
  backline/.style={-Latex, thick, draw=gray!55},
  info/.style={<->, dashed, thick, draw=gray!70},
  dashedbox/.style={draw=gray!60, dashed, rounded corners=3mm, inner sep=6mm},
  labeltext/.style={align=center}
]

\node[labeltext] (rx) {Received\\signal\\[2mm]$\mathbf{R}$};

\node[block, right=20mm of rx] (sd) {Start\\detection};
\node[block, right=18mm of sd] (ce) {Channel\\estimation};
\node[block, right=18mm of ce] (pe) {Payload\\estimation};
\node[block, right=18mm of pe] (ldpc) {LDPC\\decoder};

\node[block, below=16mm of ldpc] (td) {Tail\\detection};

\node[dashedbox, fit=(ce)(pe)(ldpc)(td),
      label={[black]above:Per-command unit processing}] (burstbox) {};

\draw[line] (rx.east) -- (sd.west);
\draw[line] (sd.east) -- node[above, black, align=center] {Start\\detected} (ce.west);
\draw[line] (ce.east) -- (pe.west);
\draw[line] (pe.east) -- (ldpc.west);

\coordinate (petap) at ($(pe.south)+(0,-10mm)$);
\draw[line] (pe.south) -- (petap) -| (td.north);

\draw[info] (ldpc.south) -- node[right, black, align=center]
  {Decoder\\information} (td.north);

\coordinate (tailjoin) at (sd.south |- td.west);
\draw[backline]
  (td.west) -- node[pos=0.80, above, black] {Tail detected}
  (tailjoin) -- (sd.south);

\coordinate (failtrack) at ($(sd.south)+(0,-42mm)$);
\draw[backline] (ldpc.east) -- ++(10mm,0) |- (failtrack)
  node[pos=0.45, below, black] {Failure}
  -| (sd.south);

\node[decision, right=20mm of ldpc] (dec) {Any CU\\decoded?};

\node[block, right=22mm of dec] (sic) {SIC};

\draw[line]
  (ldpc.east) -- node[above, black, align=center]
  {End of\\superframe} (dec.west);

\draw[line] (dec.east) -- node[above, black] {Yes} (sic.west);

\node[terminator, below=18mm of dec] (end) {};
\node[labeltext, below=4mm of end] {End};
\draw[line] (dec.south) -- node[right, black] {No} (end.north);

\coordinate (sicup) at ($(sic.north)+(0,14mm)$);
\coordinate (sdtop) at ($(sd.north)+(0,14mm)$);
\draw[backline] (sic.north) -- (sicup) -- (sdtop) -- (sd.north);

\draw[thick, gray]
  (td.east) -- ++(10mm,0);

\end{tikzpicture}
  }
  \caption{Block diagram of the proposed receiver architecture.}
  \label{fig:receiver_scheme}
\end{figure*}
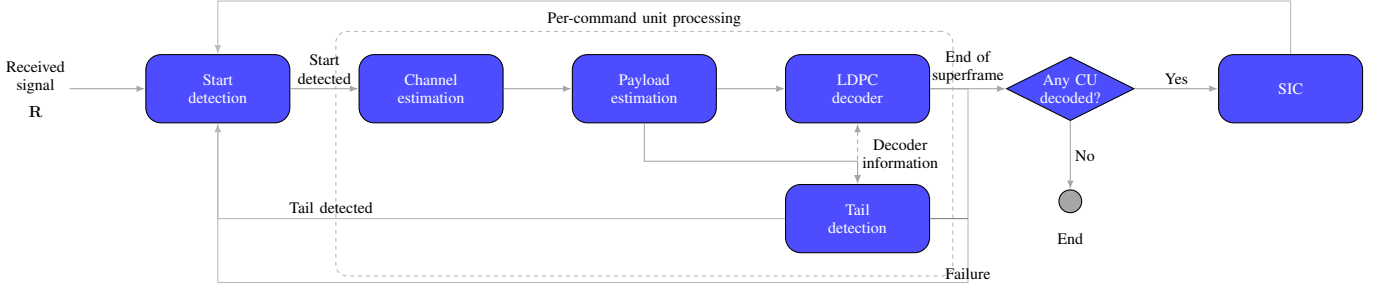

\section{Analysis of the effects of interference}\label{sec:the}

In this section, we characterize the probability of the most critical failure event from the receiver’s perspective. Since multiple \ac{SIC} rounds are permitted, the most detrimental situation arises when a command unit is prematurely declared as terminated due to the incorrect detection of a tail sequence. In this case, the receiver performs \ac{SIC} under an incorrect assumption on the signal support, subtracting only a truncated version of the actual transmission. As a result, the cancellation is erroneous and leaves residual energy that acts as interference in subsequent decoding stages. Such an event may occur under two distinct mechanisms:
\begin{enumerate}
\item a false alarm of the tail detector;
\item the tail sequence of an interfering command unit is observed before the
true tail sequence of the intended command unit.
\end{enumerate}
While the first event is associated with the tail detector design and can, to
some extent, be controlled through an appropriate trade-off with the
missed-detection probability, the second event is primarily driven by the
network traffic characteristics and the asynchronous nature of transmissions. As
such, it is of particular interest to characterize this latter event
analytically. In order to do so, we proceed as follows. First, we introduce the notation and formally define the event of interest, based on the notion of dangerous instants within a reference replica. Then, for a fixed activation time $\tau$, we derive the probability that a single interfering activation generates at least one replica whose tail aligns with such instants. Finally, we extend this result to the case of multiple interfering devices by modeling the activation process as Poisson and aggregating the contributions over all admissible activation times.

\subsection{Preliminary definitions and problem statement}

We denote the latest symbol time at which a device may attempt to communicate
with a receiving controller as $T_{\mathrm{act}}$. Accordingly, the \ac{SF}
length $T_{\mathrm{SF}}$ follows easily, since an entire \ac{VF} of any active
device must be accommodated. Under this assumption, the set of admissible device
activation times is
\[
\mathcal{A}\triangleq [0,T_{\mathrm{act}}-1],
\]
where
\[
T_{\mathrm{act}} = T_{\mathrm{SF}} - N_{\mathrm{slot}}L_{\max} + 1,
\]
and we require $T_{\mathrm{act}}\ge 1$ to ensure feasibility. We also define the number of symbols between the beginning of a replica and the beginning of its tail sequence, that is,
\[
d(\iota_j)\triangleq L_{\mathrm{pre}}+\iota_jL_{\mathrm{code}},
\]
where $\iota_j$ denotes the number of codewords in the replicas  of the $j$-th transmitting device.

We refer to as \emph{victim device} the transmitter whose replica is currently being decoded at the receiver and whose tail sequence is used to determine the termination of the command unit. All other simultaneously active devices are treated as potential interferers. So, let us consider a victim device $v$ activating at time $\tau_v$, whose
command unit contains $\iota_v\in[1,\iota_{\max}]$ codewords. Let
$\mathcal{R}_v$ denote the set of replicas transmitted by the
victim device. We focus on one victim replica, indexed by $r_v\in\mathcal{R}_v$, which
is transmitted in slot $s_{v,r_v}\in[1,N_{\mathrm{slot}}]$ with in-slot offset
$o_{v,r_v}$. The starting time of the tail sequence of the selected victim replica
is
\[
\tau_{\mathrm{tail}}^{(v,r_v)}
=
\tau_v+(s_{v,r_v}-1)L_{\max}+o_{v,r_v}+d(\iota_v).
\]

We define the \emph{dangerous set} as the set of instants at which the tail of an interfering device could be mistaken for that of the selected victim replica,
namely, 
\begin{align*}
\mathcal{D}_v
&\triangleq
\Big\{
\tau_v+(s_{v,r_v}-1)L_{\max}+o_{v,r_v}+d(j) \mid j\in[1, \iota_v-1]
\Big\}
\\
&=
\{\delta_1,\dots,\delta_{|\mathcal{D}|}\}
\end{align*}

Let $\mathcal{U}$ denote the set of all devices activated within the
admissible activation window $\mathcal{A}$. For a fixed victim replica
$(v,r_v)$, define the set of interfering devices as
\[
\begin{aligned}
\mathcal{U}_{\mathrm{int}}
\triangleq &
\Big\{
u \in \mathcal{U}:\;
\exists\, r_u \in \mathcal{R}_u \ \text{such that}
\\
&
[\tau^{(u,r_u)}_{\mathrm{start}},\,\tau^{(u,r_u)}_{\mathrm{end}}]
\cap
[\tau^{(v,r_v)}_{\mathrm{start}},\,\tau^{(v,r_v)}_{\mathrm{end}}]
\neq \emptyset
\Big\}.
\end{aligned}
\]
where $\tau^{(u,r_u)}_{\mathrm{start}}$ and $\tau^{(u,r_u)}_{\mathrm{end}}$ denote,
respectively, the starting and ending symbol times of the $r$-th replica
transmitted by device $u$, and $\tau^{(v,r_v)}_{\mathrm{start}}$ and
$\tau^{(v,r_v)}_{\mathrm{end}}$ are defined analogously for the selected victim
replica.

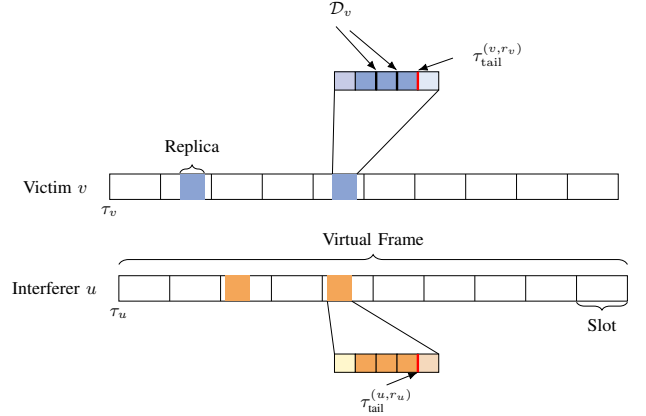
\begin{figure}[tb!]
  \centering
  \resizebox{0.95\linewidth}{!}{%
    \begin{tikzpicture}[
    x=0.9cm,
    y=0.9cm,
    >=Latex,
    every node/.style={font=\small}
]

\def\L{10}
\def\H{0.5}
\def\rwv{0.48}
\def\rwu{0.48}

\def\roffv{0.38}
\def\roffu{0.10}

\def\xv0{0.00}
\def\xu0{0.18}

\tikzset{
    vrep/.style={fill=cyan!35!blue!45},
    urep/.style={fill=orange!75!brown!70}
}

\def\yv{0}
\def\yu{-2}

\def\slotA{1}
\def\slotB{4}
\def\slotC{2}
\def\slotU{4}

\pgfmathsetmacro{\vxA}{\xv0+\slotA+\roffv}
\pgfmathsetmacro{\vxB}{\xv0+\slotB+\roffv}
\pgfmathsetmacro{\uxA}{\xu0+\slotC+\roffu}
\pgfmathsetmacro{\uxB}{\xu0+\slotU+\roffu}

\node[anchor=east] at (\xv0-0.3,\yv+0.25) {Victim $v$};

\draw (\xv0,\yv) rectangle ({\xv0+\L},\yv+\H);
\foreach \x in {1,...,9}{
    \draw ({\xv0+\x},\yv) -- ({\xv0+\x},\yv+\H);
}

\fill[vrep] (\vxA,\yv) rectangle ({\vxA+\rwv},\yv+\H);
\fill[vrep] (\vxB,\yv) rectangle ({\vxB+\rwv},\yv+\H);

\node[below] at (\xv0,\yv) {$\tau_v$};

\draw[decorate,decoration={brace,amplitude=4pt}]
(\vxA,\yv+\H+0.02) -- ({\vxA+\rwv},\yv+\H+0.02)
node[midway,above=5pt] {Replica};

\node[anchor=east] at (\xu0-0.3,\yu+0.25) {Interferer $u$};

\draw (\xu0,\yu) rectangle ({\xu0+\L},\yu+\H);
\foreach \x in {1,...,9}{
    \draw ({\xu0+\x},\yu) -- ({\xu0+\x},\yu+\H);
}

\fill[urep] (\uxA,\yu) rectangle ({\uxA+\rwu},\yu+\H);
\fill[urep] (\uxB,\yu) rectangle ({\uxB+\rwu},\yu+\H);

\node[below] at (\xu0,\yu) {$\tau_u$};

\draw[decorate,decoration={brace,amplitude=5pt}]
(\xu0,\yu+\H+0.22) -- ({\xu0+\L},\yu+\H+0.22)
node[midway,above=6pt] {Virtual Frame};

\draw[decorate,decoration={brace,mirror,amplitude=4pt}]
({\xu0+9},\yu-0.02) -- ({\xu0+10},\yu-0.02)
node[midway,below=5pt] {Slot};

\def\zx{4.42}
\def\zy{2.15}
\def\zw{2.05}
\def\zh{0.36}

\coordinate (ZL) at (\zx,\zy);
\coordinate (ZR) at ({\zx+\zw},\zy);

\draw (\vxB,\yv+\H) -- (ZL);
\draw ({\vxB+\rwv},\yv+\H) -- (ZR);

\foreach \k in {0,1,2,3,4}{
    \pgfmathsetmacro{\xleft}{\zx+\k*\zw/5}
    \pgfmathsetmacro{\xright}{\zx+(\k+1)*\zw/5}

    \ifnum\k=0
        \fill[cyan!20!blue!20] (\xleft,\zy) rectangle (\xright,{\zy+\zh});
    \else\ifnum\k=4
        \fill[cyan!50!blue!10] (\xleft,\zy) rectangle (\xright,{\zy+\zh});
    \else
        \fill[vrep] (\xleft,\zy) rectangle (\xright,{\zy+\zh});
    \fi\fi

    \draw (\xleft,\zy) rectangle (\xright,{\zy+\zh});
}

\def\zux{4.42}
\def\zuy{-3.4}

\coordinate (ZLu) at (\zux,\zuy+\zh);
\coordinate (ZRu) at ({\zux+\zw},\zuy+\zh);

\draw (\uxB,\yu) -- (ZLu);
\draw ({\uxB+\rwu},\yu) -- (ZRu);

\foreach \k in {0,1,2,3,4}{
    \pgfmathsetmacro{\xleft}{\zux+\k*\zw/5}
    \pgfmathsetmacro{\xright}{\zux+(\k+1)*\zw/5}

    \ifnum\k=0
        \fill[orange!30!yellow!20] (\xleft,\zuy) rectangle (\xright,{\zuy+\zh});
    \else\ifnum\k=4
        \fill[orange!50!brown!30] (\xleft,\zuy) rectangle (\xright,{\zuy+\zh});
    \else
        \fill[urep] (\xleft,\zuy) rectangle (\xright,{\zuy+\zh});
    \fi\fi

    \draw (\xleft,\zuy) rectangle (\xright,{\zuy+\zh});
}

\pgfmathsetmacro{\xfiveu}{\zux+4*\zw/5}

\draw[line width=1.2pt,red] (\xfiveu,\zuy) -- (\xfiveu,\zuy+\zh);

\node at ({\zux+0.5*\zw}, {\zuy-0.55}) {$\tau_{\text{tail}}^{(u,r_u)}$};
\draw[->] ({\zux+0.5*\zw}, {\zuy-0.25}) -- (\xfiveu,\zuy);

\pgfmathsetmacro{\xthree}{\zx+2*\zw/5}
\pgfmathsetmacro{\xfour}{\zx+3*\zw/5}
\pgfmathsetmacro{\xfive}{\zx+4*\zw/5}

\draw[line width=1.2pt] (\xthree,\zy) -- (\xthree,\zy+\zh);
\draw[line width=1.2pt] (\xfour,\zy) -- (\xfour,\zy+\zh);
\draw[line width=1.2pt,red] (\xfive,\zy) -- (\xfive,\zy+\zh);

\node[align=center] at (4.55,3.75) {$\mathcal{D}_v$};
\draw[->] (4.45,3.40) -- (\xthree,\zy+\zh+0.03);
\draw[->] (4.72,3.32) -- (\xfour,\zy+\zh+0.03);

\node[anchor=west] at (7.00,2.85) {$\tau_{\mathrm{tail}}^{(v,r_v)}$};
\draw[->] (6.75,2.78) -- (\xfive,\zy+\zh+0.03);

\end{tikzpicture}
  }
  \caption{Asynchronous \acp{VF} of a victim $v$ and an interferer $u$.}
  \label{fig:help}
\end{figure}

Let $u\in\mathcal{U}_{\mathrm{int}}$ denote an interfering device.  For each $\delta\in\mathcal{D}_v$,
define the random variable
\[
N_\delta
\triangleq
\sum_{u\in\mathcal{U}_{\mathrm{int}}}
\ \sum_{r_u\in\mathcal{R}_u}
\llbracket \tau_{\mathrm{tail}}^{(u,r_u)}=\delta\rrbracket,
\]
where $\tau_{\mathrm{tail}}^{(u,r_u)}$ denotes the starting time of the tail
sequence of the $r$-th replica transmitted by the $u$-th interfering device and  $\llbracket \cdot \rrbracket$ denotes the Iverson bracket, defined in Section \ref{subsec:not}.

For the interfering device $u$, the activation
time is $\tau_u$, and $\iota_u\in[1,\iota_{\max}]$ denotes the
number of codewords in the replicas it transmits. Given the $r$-th replica,  the slot index is $s_{u,r}\in[1,N_{\mathrm{slot}}]$, and 
the in-slot offset is $o_{u,r}$. The starting
time of the tail sequence of the $r$-th replica of device $u$ is therefore
\[
\tau_{\mathrm{tail}}^{(u,r_u)}
=
\tau_u+(s_{u,r}-1)L_{\max}+o_{u,r}+d(\iota_u).
\]

A \emph{tail confusion} event occurs if there exists at least one interfering
device (say the $u$-th) producing a replica (say $r_u$) whose tail starting
time is $\tau_{\mathrm{tail}}^{(u,r_u)}$, and this time coincides with a dangerous instant of the
selected victim replica. The structure of the dangerous set and its relation to interfering replicas
is illustrated in Fig.~\ref{fig:help}. The upper part of the figure shows a
selected replica of the victim device $v$, together with the dangerous
instants $\mathcal{D}_v$ within the replica. The lower part depicts an
interfering device $u$. The enlarged parts highlight the possible alignment between the dangerous instants of
the victim replica and the tail position of the interferer: if $\tau_{\mathrm{tail}}^{(u,r_u)} \in \mathcal{D}_v$, we have a tail confusion event. Formally, we define
\[
H
\triangleq
\Big\{
\exists\, (u,r_u)\ \text{with } u\in\mathcal{U}_{\mathrm{int}},\ r\in\mathcal{R}_u
:\ \tau_{\mathrm{tail}}^{(u,r_u)} \in \mathcal{D}_v
\Big\}.
\]
Equivalently, we have
\[
H
=
\Big\{
\exists\,\delta\in\mathcal{D}_v:\ N_\delta \ge 1
\Big\},
\]
and
\[
\mathbb{P}(H)
=
1-\mathbb{P}\Big(
\forall\,\delta\in\mathcal{D}_v,\ N_\delta=0
\Big).
\]

In the following, we will derive an expression for $\mathbb{P}(H)$.

\subsection{Dangerous-activation probability for a fixed activation time}

Let us fix the dangerous set $\mathcal{D}_v$ and the selected victim replica $(v,r_v)$. Let the selected victim replica start at
\[
\tau_{\mathrm{start}}^{(v,r_v)}\triangleq \tau_v+(s_{v,r_v}-1)L_{\max}+o_{v,r_v},
\]
and let us define its (slot-budget) occupancy interval as
\[
I_v \triangleq \big[\tau_{\mathrm{start}}^{(v,r_v)},\ \tau_{\mathrm{start}}^{(v,r_v)}+L_{\max}-1\big].
\]
An activation at time $\tau\in\mathcal{A}$ can generate an overlapping replica
with $(v,r_v)$ only if it belongs to the set
\[
\mathcal{A}_{\mathrm{int}}
\triangleq
\Big\{\tau\in\mathcal{A}:\ \exists\, s\in[1,N_{\mathrm{slot}}],\ \exists\, o,\ \exists\,\iota
\text{ such that }\]
\[
\big[\tau+(s-1)L_{\max}+o,\ \tau+sL_{\max}-1\big]\cap I_v\neq\emptyset\Big\}.
\]
Accordingly, in the following we only consider $\tau\in\mathcal{A}_{\mathrm{int}}$.

Let us fix an activation time $\tau_u\in\mathcal{A}_{\mathrm{int}}$ and define the event
that at least one of the $N_{\mathrm{rep}}$ replicas generated by such activation
produces a tail starting at a dangerous instant of the selected victim replica as
\begin{equation}
E_{\mathcal{D}}(\tau_u)\triangleq
\Big\{
\exists\, r\in\mathcal{R}_u:\ \tau_{\mathrm{tail}}^{(u,r_u)}\in\mathcal{D}_v
\Big\}.
\label{eq:edtu}
\end{equation}

Let us fix $\iota\in[1,\iota_{\max}]$, the activation time $\tau_u$, and the slot index $s_{u,r}=s\in[1,N_{\mathrm{slot}}]$. Under these conditions, the in-slot offset $o_{u,r}$ is uniformly distributed over the integer interval
\[
[0,O_{\max}]=[0,\,L_{\max}-d(\iota)-L_{\mathrm{tail}}],
\]
and offset realizations are independent across replicas. Moreover, the $N_{\mathrm{rep}}$ slots are selected uniformly at random without replacement from $[1,N_{\mathrm{slot}}]$.

\begin{Lem}\label{lem:pdelta}
Fix $\tau_u\in\mathcal{A}_{\mathrm{int}}$, a slot index $s\in[1,N_{\mathrm{slot}}]$,
and a command-unit length $\iota\in[1,\iota_{\max}]$. Under the uniform-offset
assumption, for any integer time $\delta$,
\begin{align}
p_{\delta}&(\tau_u,s,\iota)
\triangleq
\mathbb{P}\!\left(
\tau_{\mathrm{tail}}^{(u,r_u)}=\delta
\ \middle|\ 
\tau_u,\ s_{u,r}=s,\ \iota_u=\iota
\right)\nonumber\\
&=
\frac{
\llbracket
0\le \delta-\tau_u-(s-1)L_{\max}-d(\iota)
\le O_{\max}
\rrbracket
}{
L_{\max}-d(\iota)-L_{\mathrm{tail}}+1
}.
\label{eq:pdelta_lemma}
\end{align}
\end{Lem}
\begin{IEEEproof}
Fix $(\tau_u,s_{u,r}=s,\iota_u=\iota)$. 
From the definition of tail starting time we have
\[
\tau_{\mathrm{tail}}^{(u,r_u)}
=
\tau_u+(s-1)L_{\max}+o_{u,r}+d(\iota).
\]
Hence the event $\{\tau_{\mathrm{tail}}^{(u,r_u)}=\delta\}$ occurs if and only if
\[
o_{u,r}
=
\delta-\tau_u-(s-1)L_{\max}-d(\iota)
\triangleq o^\star .
\]
Therefore
\[
\mathbb{P}\!\left(
\tau_{\mathrm{tail}}^{(u,r_u)}=\delta
\ \middle|\
\tau_u,\ s_{u,r}=s,\ \iota_u=\iota
\right)
=
\mathbb{P}(o_{u,r}=o^\star).
\]

By assumption,  the offset $o_{u,r}$ is uniformly
distributed over the integer set
$$[0,O_{\max}],
\qquad
O_{\max}=L_{\max}-d(\iota)-L_{\mathrm{tail}}.
$$
Hence the number of admissible offset values is $O_{\max}+1$, and
\[
\mathbb{P}(o_{u,r}=k)=\frac{1}{O_{\max}+1},
\qquad
k\in[0,O_{\max}].
\]

If $o^\star\notin[0,O_{\max}]$, the event is impossible and the probability is
zero. Otherwise, if $o^\star\in[0,O_{\max}]$, exactly one offset value produces
$\tau_{\mathrm{tail}}^{(u,r_u)}=\delta$, yielding
\[
\mathbb{P}(o_{u,r}=o^\star)=\frac{1}{O_{\max}+1}.
\]

Combining the two cases and recalling that
$O_{\max}=L_{\max}-d(\iota)-L_{\mathrm{tail}}$ yields  \eqref{eq:pdelta_lemma}.
\end{IEEEproof}

\begin{Cor}\label{cor:pD}
For fixed $(\tau_u,s,\iota)$, the per-slot probability that a replica transmitted
in slot $s$ hits the dangerous set $\mathcal{D}_v$ satisfies
\begin{align}
p_{\mathcal{D}}(\tau_u,s,\iota)
&\triangleq
\mathbb{P}\!\left(
\tau_{\mathrm{tail}}^{(u,r_u)}\in\mathcal{D}_v
\ \middle|\ 
\tau_u,\ s_{u,r}=s,\ \iota_u=\iota
\right)\nonumber\\
&=
\sum_{\delta\in\mathcal{D}_v} p_{\delta}(\tau_u,s,\iota),
\label{eq:pD_cor}
\end{align}
with $p_{\delta}(\tau_u,s,\iota)$ given by Lemma~\ref{lem:pdelta}.
\end{Cor}
\begin{IEEEproof}
Since $\tau_{\mathrm{tail}}^{(u,r_u)}$ takes a single integer value, the events
$\{\tau_{\mathrm{tail}}^{(u,r_u)}=\delta\}$ are mutually exclusive over distinct
$\delta\in\mathcal{D}_v$, hence the probability of membership in $\mathcal{D}_v$
equals the sum of the point probabilities.
\end{IEEEproof}

\begin{The}\label{thm:nohit}
Fix $\tau_u\in\mathcal{A}_{\mathrm{int}}$ and $\iota\in[1,\iota_{\max}]$.  Then the
probability that none of the $N_{\mathrm{rep}}$ replicas has its tail starting time in
$\mathcal{D}_v$ equals
\begin{align}
\mathbb{P}(\textnormal{no hit in }\mathcal{D}_v \mid \tau_u,\iota)
&=\nonumber\\
\frac{1}{\binom{N_{\mathrm{slot}}}{N_{\mathrm{rep}}}}
\sum_{\substack{\mathcal{S}\subseteq[1,N_{\mathrm{slot}}]\\|\mathcal{S}|=N_{\mathrm{rep}}}}
& \prod_{s\in\mathcal{S}}\big(1-p_{\mathcal{D}}(\tau_u,s,\iota)\big).
\label{eq:nohit_thm}
\end{align}
\end{The}
\begin{IEEEproof}
Let $S$ denote the random subset of slots chosen by the activation, where $S$ is
uniform over all $\binom{N_{\mathrm{slot}}}{N_{\mathrm{rep}}}$ subsets of size
$N_{\mathrm{rep}}$. Let $\mathcal{S}$ be a fixed realization of $S$. For each $s\in\mathcal{S}$ there is exactly one replica transmitted in slot
$s$, and $p_{\mathcal{D}}(\tau_u,s,\iota)$ is the probability that this replica
hits $\mathcal{D}_v$. Since offsets are independent across replicas, the
corresponding ``no-hit'' events are independent for all $s\in\mathcal{S}$, hence
\[
\mathbb{P}(\text{no hit in }\mathcal{D}_v \mid \tau_u,\iota,S=\mathcal{S})
=
\prod_{s\in\mathcal{S}}\big(1-p_{\mathcal{D}}(\tau_u,s,\iota)\big).
\]
Averaging over all possible realizations of $S$, which are equiprobable, yields
\eqref{eq:nohit_thm}.
\end{IEEEproof}

Fix $\tau_u\in\mathcal{A}_{\mathrm{int}}$ and $\iota\in[1,\iota_{\max}]$.  Then
\begin{align}
q_{\mathcal{D}}(\tau_u\mid \iota)
&\triangleq
\mathbb{P}\!\left(E_{\mathcal{D}}(\tau_u)\mid \tau_u,\iota\right)\nonumber\\
&=
1-\mathbb{P}(\textnormal{no hit in }\mathcal{D}_v\mid \tau_u,\iota),
\label{eq:qD_cor}
\end{align}
where $E_{\mathcal{D}}(\tau_u)$ is defined in \eqref{eq:edtu} and $\mathbb{P}(\textnormal{no hit in }\mathcal{D}_v\mid \tau_u,\iota)$ is given
by Theorem~\ref{thm:nohit}. Also, assume $\iota_u\sim\mathrm{Unif}[1,\iota_{\max}]$ independently of $\tau_u$. Then
\[
q_{\mathcal{D}}(\tau_u)
\triangleq
\mathbb{P}\!\left(E_{\mathcal{D}}(\tau_u)\mid \tau_u\right)
=
\frac{1}{\iota_{\max}}
\sum_{\iota=1}^{\iota_{\max}} q_{\mathcal{D}}(\tau_u\mid \iota).
\]

The quantity $q_{\mathcal{D}}(\tau)$ characterizes the probability that a single
activation at time $\tau$ generates at least one replica whose tail hits
$\mathcal{D}_v$. The tail confusion event $H$ occurs if at least one such
``dangerous'' activation takes place at some $\tau\in\mathcal{A}_{\mathrm{int}}$. In the following section, we compute $\mathbb{P}(H)$ under a Poisson activation model.

\subsection{Tail confusion probability under Poisson activations}

\begin{The}\label{thm:thinning}
Assume that, for each symbol time $\tau\in\mathcal{A}_{\mathrm{int}}$, the number
of interfering device activations $M_\tau$ is Poisson distributed with mean $\lambda$,
independently across $\tau$. 
Let us define the number of ``dangerous'' activations at time $\tau$ as
\begin{equation}
C_\tau \triangleq \sum_{\ell=1}^{M_\tau} \llbracket E_{\mathcal{D}}^{(\ell)}(\tau)\rrbracket,
\label{eq:dangact}
\end{equation}
where, conditioned on $\tau$, the events $\{E_{\mathcal{D}}^{(\ell)}(\tau)\}$ are
i.i.d. Bernoulli variables with parameter $q_{\mathcal{D}}(\tau)$ and are independent of
$M_\tau$. Then
\[
C_\tau \sim \mathrm{Poisson} \big(\lambda q_{\mathcal{D}}(\tau)\big),
\]
and the random variables $\{C_\tau\}_{\tau\in\mathcal{A}_{\mathrm{int}}}$ are
independent across $\tau$.
\end{The}
\begin{IEEEproof}
Fix $\tau\in\mathcal{A}_{\mathrm{int}}$ and let $G_{C_\tau}(s)=\mathbb{E}[s^{C_\tau}]$
denote the \ac{PGF} of $C_\tau$. We have
\[
  G_{M_\tau}(s)=\exp\bigl(\lambda(s-1)\bigr),
  \qquad
  G_B(s)=1+q_{\mathcal{D}}(\tau)(s-1),
\]
which are the \ac{PGF}s of a $\mathrm{Poisson}(\lambda)$ and a
$\mathrm{Bernoulli}(q_{\mathcal{D}}(\tau))$ random variable, respectively.
By the binomial structure induced by~\eqref{eq:dangact},
$\mathbb{E}[s^{C_\tau}\mid M_\tau]=G_B(s)^{M_\tau}$,
so by the law of total expectation and the \ac{PGF} composition formula,
\[
  G_{C_\tau}(s)
  =G_{M_\tau}\!\bigl(G_B(s)\bigr)
  =\exp\!\left(\lambda q_{\mathcal{D}}(\tau)(s-1)\right),
\]
which is the \ac{PGF} of a $\mathrm{Poisson}(\lambda q_{\mathcal{D}}(\tau))$
random variable, so $C_\tau\sim\mathrm{Poisson}(\lambda q_{\mathcal{D}}(\tau))$.

It remains to prove mutual independence. Let
$\tau_1,\ldots,\tau_n\in\mathcal{A}_{\mathrm{int}}$ be distinct. Again by the law of total expectation,
\[
  \mathbb{E}\!\left[\prod_{k=1}^n s_k^{C_{\tau_k}}\right]
  =\mathbb{E}\!\left[
     \mathbb{E}\!\left[
       \prod_{k=1}^n s_k^{C_{\tau_k}}
       \,\middle|\,
       M_{\tau_1},\ldots,M_{\tau_n}
     \right]
   \right].
\]
Conditioned on $M_{\tau_1},\ldots,M_{\tau_n}$, the counts $C_{\tau_1},\ldots,C_{\tau_n}$
are independent, since each $C_{\tau_k}$ is formed from Bernoulli random variables
associated exclusively with symbol time $\tau_k$. The inner expectation therefore
factorizes, and since $M_{\tau_1},\ldots,M_{\tau_n}$ are independent, taking the
outer expectation gives
\begin{align*}
  \mathbb{E}\!\left[\prod_{k=1}^n s_k^{C_{\tau_k}}\right]
  &=\prod_{k=1}^n
    \mathbb{E}\!\left[
      G_B(s_k)^{M_{\tau_k}}
    \right]
  =\prod_{k=1}^n G_{C_{\tau_k}}(s_k).
\end{align*}
The joint \ac{PGF} factors into the product of the marginals, so
$\{C_\tau\}_{\tau\in\mathcal{A}_{\mathrm{int}}}$ are mutually independent.
\end{IEEEproof}

\begin{The}\label{thm:PH}
Let us assume that, for each symbol time $\tau\in\mathcal{A}_{\mathrm{int}}$, the number
of device activations $M_\tau$ is Poisson distributed with mean $\lambda$,
independently across $\tau$. Let us define the total number of dangerous activations over
$\mathcal{A}_{\mathrm{int}}$ as
\[
C \triangleq \sum_{\tau\in\mathcal{A}_{\mathrm{int}}} C_\tau,
\]
where $C_\tau$ is the number of dangerous activations at time $\tau$. Then $C$ is
Poisson distributed with mean
\[
\Lambda_H \triangleq \lambda\sum_{\tau\in\mathcal{A}_{\mathrm{int}}}
q_{\mathcal{D}}(\tau).
\]
Moreover, the tail confusion event $H$ occurs if and only if $C\ge 1$, with probability
\[
\mathbb{P}(H)
=
1-\exp\left(-\lambda\sum_{\tau\in\mathcal{A}_{\mathrm{int}}}
q_{\mathcal{D}}(\tau)\right).
\]
\end{The}
\begin{IEEEproof}
By Theorem~\ref{thm:thinning}, for each $\tau\in\mathcal{A}_{\mathrm{int}}$ the
number of dangerous activations $C_\tau$ is Poisson distributed with mean
$\lambda q_{\mathcal{D}}(\tau)$, and the random variables $\{C_\tau\}$ are
independent across $\tau$. Hence, their sum
\[
C=\sum_{\tau\in\mathcal{A}_{\mathrm{int}}} C_\tau
\]
is Poisson distributed with mean
$\Lambda_H=\lambda\sum_{\tau\in\mathcal{A}_{\mathrm{int}}} q_{\mathcal{D}}(\tau)$.

By definition, a tail confusion occurs if and only if at least one dangerous
activation is present, i.e., if $C\ge 1$. Since $C$ is Poisson distributed,
\[
\mathbb{P}(H)=\mathbb{P}(C\ge 1)=1-\mathbb{P}(C=0)=1-e^{-\Lambda_H},
\]
which yields the stated result.
\end{IEEEproof}

\begin{table}[t]
   \caption{CNN complexity.}
  \label{tab:CNN_compl}
  \centering
  \scriptsize
    \begin{tabular}{l|l|l|l}
    \toprule
    \textbf{Layer type} & \textbf{Layer size} & \textbf{Activation} & \textbf{FLOPs} \\
    \midrule
    Input$^*$ $\bm Y_1$ & $3 \times 2 L_{\text{pre}}$ & & \\
    2D convolutional & $8$@$2\times 7$ & ReLu & 116,000 \\
    Fully connected & $130$ & ReLu & 1,040,000 \\
    Fully connected & $65$ & ReLu & 16,965\\
    \midrule
    Input$^*$ $\bm Y_2$ & $(i_{max}+1) \times L_{\text{tail}}$ & & \\
    Fully connected & $130$ & ReLu & 699,010 \\
    Fully connected & $65$ & ReLu & 16,965\\
    \midrule
    Fully connected & $130$ & ReLu & 33,930 \\
    Fully connected & $65$ & ReLu & 16,965\\
    Fully connected & $2$ & ReLu & 262\\
    Output & $2$ & Sigmoid & 8 \\
    
    \bottomrule
    \end{tabular}

    \vspace{1mm}
    \begin{tabular}{p{0.9\linewidth}}
        \footnotesize $^*$ The input dimensions depend on the lengths of the preamble and tail sequences, which are assumed to be equal in our framework, i.e., $L_{\text{pre}} = L_{\text{tail}}$.
    \end{tabular}
\end{table}

\section{Boundary detection}\label{sec:bound}

We propose a unified \ac{DL}-based framework for the joint detection of start and tail sequences, based solely on the received signal. Both detection tasks are handled by a single \ac{NN}, relying only on signal-level features for start sequence detection, while also leveraging decoder information for tail detection. The proposed solution adopts a multi-branch, multi-label \ac{CNN} architecture. As illustrated in Fig.~\ref{fig:CNNarch}, the network consists of two parallel branches processing the inputs $\bm{Y}_1$ and $\bm{Y}_2$, respectively, which are subsequently merged into a shared fully connected stage. This final stage comprises three fully connected layers with \ac{ReLU} activation functions and $130$, $65$, and $2$ neurons, respectively. 
The final layer of the network utilizes a sigmoid activation function, allowing each output to independently represent the probability of the presence of the start ($A$) and tail ($B$) sequences. This configuration supports multi-label classification, as the labels $[A, B] \in \{0,1\}^2$ are not mutually exclusive.

\begin{figure}[t]
    \centering 
    \includegraphics[width=0.4\paperwidth]{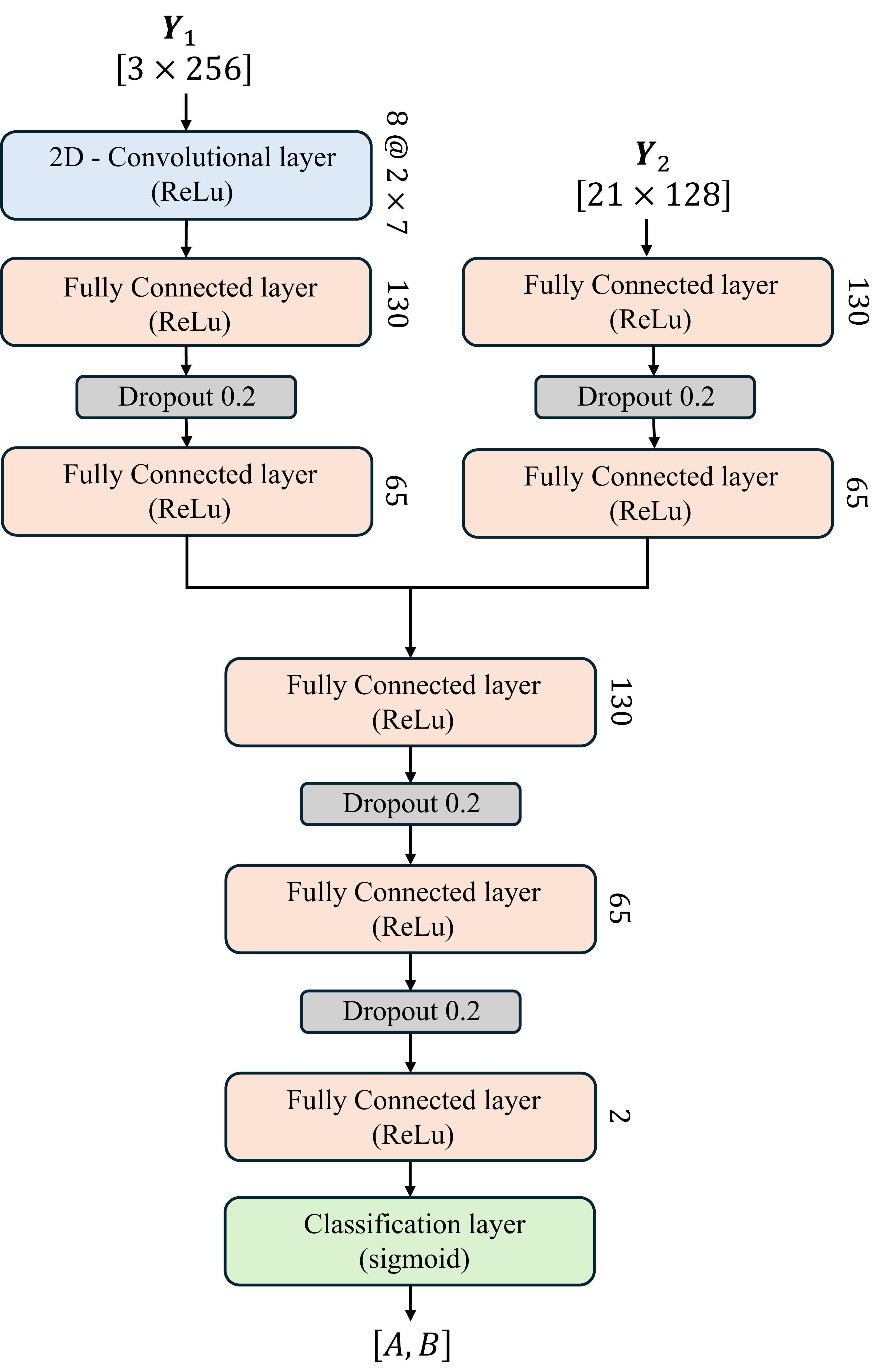} 
    \caption{A schematic representation of the architecture of the proposed CNN for start and tail detection. The size of each layer is specified.}
    \label{fig:CNNarch}
\end{figure}

The network inputs are derived from segments of length $L_\text{pre}=L_\text{tail}$ extracted from the received signal buffer, each belonging to one of the following classes: 
\begin{itemize}
    \item H$_0$: $[A,B]=[0,0]$, noise or non-informative fragments;
    \item H$_1$: $[A,B]=[1,0]$, start sequence; 
    \item H$_2$: $[A,B]=[0,1]$, tail sequence; 
    \item H$_3$: $[A,B]=[0,0]$, codeword;\footnote{Although H$_0$ and H$_3$ share the same label $[0,0]$, they are considered separate classes to distinguish non-informative fragments from valid codewords. While it is theoretically possible to merge them under the label $[0,0]$, the construction of the input $\bm{Y}_2$ differs in each case.}
    \item H$_4$: $[A,B]=[1,1]$, overlapping start  and tail.
\end{itemize}

Even with a single training process, this network can perform two tasks: the first one focuses on frame synchronization through the detection of the start sequence; then, once synchronization is achieved, the network is employed for tail detection. In this case, the second branch of the \ac{DL} model leverages information obtained from the decoder.
It is worth noting that, even during the tail detection phase, the network is able to identify start sequences that are exactly overlapping a codeword or a tail sequence, thanks to the multi-label classification mechanism.

\subsection{Start sequence detection}
The goal is to verify whether $L_\text{pre}$ consecutive received samples correspond to the known start  sequence. 
The input to the first branch of the \ac{DL} model is the observation matrix \( \bm{Y}_1 \in \mathbb{R}^{3 \times (2 L_{\text{pre}})} \), which serves as a feature map constructed by concatenating three components: 
(i) the real and imaginary parts of the known start  sequence, 
(ii) the real and imaginary parts of the \ac{MRC} estimate of the received samples \( \hat{\mathbf{x}}^\top \in \mathbb{C}^{L_{\text{pre}}} \), computed using~\eqref{eq:mrc_estimation} and the channel estimate from~\eqref{eq:ch_estimate}, and 
(iii) the real and imaginary parts of the known tail sequence. 
The resulting matrix is structured as:
\begin{align}
    \bm{Y}_1 & = 
    \setlength\arraycolsep{2pt}
    \begin{bmatrix}
        \Re(\bf{x}_\text{pre}) & \Im(\bf{x}_\text{pre}) \\
        \Re(\hat{\bf{x}}) & \Im(\hat{\bf{x}}) \\
        \Re(\bf{x}_\text{tail}) & \Im(\bf{x}_\text{tail})
    \end{bmatrix}  
    \,.
    \label{eq:Y1}
\end{align}
To process the two-dimensional input feature map, we exploit a convolutional branch as shown in Fig.~\ref{fig:CNNarch}, consisting of a convolutional layer with $8$ filters of size  $2 \times 7$ (with no padding and stride $1$), followed by two fully connected layers of $130$ and $65$ neurons, respectively. All three layers use the \ac{ReLU} activation function.
The input to the second branch is the observation matrix $\bm{Y}_2 \in \mathbb{R}^{(i_{\max}+1) \times L_\text{tail}}$, defined as 
\begin{align}
    \bm{Y}_2 & = 
    \setlength\arraycolsep{2pt}
    \begin{bmatrix}
        \bm{L} \\
        \bm{f}
    \end{bmatrix}.
    \label{eq:Y2}
\end{align}
In the case of start detection, the matrix $\bm{L}\in\mathbb{R}^{i_{\max}\times L_\text{tail}}$ is filled with the correlation value computed between the current input and the known preamble sequence, while the vector $\bm{f} \in\mathbb{R}^{1\times L_\text{tail}} $ is set to all zeros. The second branch is specifically designed for use during tail detection, as discussed next, and consists of two fully connected layers with \ac{ReLU} activation functions, containing $130$ and $65$ neurons, respectively.
Start sequence detection is addressed as a binary classification task in which the output label $A$ takes the value $1$ to denote the presence of the start sequence, and $0$ to indicate its absence.

\subsection{Tail detection}
The objective is to determine whether $L_\text{tail}$ consecutive received samples match the known tail sequence. 
At this stage, since the packet has already been synchronized, we operate on blocks of $L_\text{tail}$ samples, each of which is known to represent either a codeword or a tail sequence. These blocks are fed to both the decoder and the \ac{CNN}, the latter being a support for tail detection.
The input to the first branch, $\bm{Y}_1$, is obtained as per \eqref{eq:Y1}, while the input of the second branch, $\bm{Y}_2$, described in \eqref{eq:Y2}, exploits additional information extracted from the decoding process. Given a code with block length $L_\text{code}=L_\text{tail}$, we build $\bm{L}\in\mathbb{R}^{i_{\max}\times L_\text{code}}$ as a matrix that contains the \acp{LLR} calculated by the decoder during each of the $i_{\max}$ iterations, for each of the $L_\text{code}$ code bits. The decoder is run for a fixed number of iterations $i_{\max}$, without early stopping, so as to ensure a consistent input dimensionality for the network.
Instead, $\bm{f} \in \{ \mathbf{0}_{L_\text{code}}, \mathbf{1}_{L_\text{code}} \}$ is a flag vector indicating whether the decoder converged (i.e., it returned a codeword) or not.
Similarly to the start sequence, tail detection is formulated as a binary classification problem, with the output label $B$ set to $1$ if a tail is detected, and $0$ otherwise.

\begin{Ass}
The proposed framework for tail detection is not restricted to \ac{BP} decoding of \ac{LDPC} codes. 
More generally, it applies to any coding scheme for which a \ac{SISO} incomplete \cite[Ch.~1]{Blahut2003} decoder is available, that is, a decoder that returns a real-valued string as output that does not necessarily correspond to a valid codeword.
Moreover, the approach naturally extends to concatenated coding schemes, where a 
\ac{CRC} is employed as an outer code, enabling error detection after inner decoding.
\end{Ass}

\subsection{CNN complexity}

We evaluate the computational complexity of the CNN in terms of \acp{FLOP}; we assume that real addition, subtraction, and multiplication count as a single FLOP, while division and exponential operations count as 4 and 8 \acp{FLOP}, respectively. For the complex addition and subtraction operations we consider two \acp{FLOP}, while for complex multiplication we consider six \acp{FLOP}.

The number of \acp{FLOP} of a convolutional layer is given by
\begin{align}
    C_{\mathrm{cv}}  &= 2 \cdot N_{\mathrm{cv}} \cdot F_{\mathrm{cv}} \cdot G_{\mathrm{cv}} \cdot D_{\mathrm{cv}},
\end{align}
where $N_{\mathrm{cv}}$, $F_{\mathrm{cv}}$, $G_{\mathrm{cv}}$, and $D_{\mathrm{cv}}$ represent the number of convolution filters, size of the filter, number of channels, and output shape, respectively. The output shape $D_{\mathrm{cv}}$ is expressed as $(I-J+2 \cdot P)/S+1$, where $I$, $J$, $P$, and $S$ specify the input size, filter size, padding, and stride. The \ac{ReLU} activation function applied to the output of the convolutional layer, results in
\begin{align}
    C_{\mathrm{ReLU}} &=  N_{\mathrm{cv}} \cdot D_{\mathrm{cv}} \,.
\end{align}
The number of \acp{FLOP} in a fully-connected layer can be expressed as
\begin{align}\label{eq:complex_FC}
    C_{\mathrm{FC}} &=   2 \cdot \mathrm{in} \cdot \mathrm{out} + \mathrm{out}\, .
\end{align}
where $\mathrm{in}$ and $\mathrm{out}$ are the input and output dimensions of the layer.
The sigmoid function and thresholding in the last layer yield $4$ FLOPs for each output label. Details about the complexity of the \ac{CNN} are reported in Table~\ref{tab:CNN_compl}; the overall complexity results in approximately $1.940\cdot 10^6$ \acp{FLOP}.
Considering an embedded edge platform such as the \textit{NVIDIA Jetson Orin Nano 4GB}, which achieves up to $1.04$\,TFLOPs/s in FP32 operations, the inference time can be estimated as
\begin{align*}
t_{\text{run}} \approx \frac{1.940 \times 10^6\ \text{\acp{FLOP}}}{1.04 \times 10^{12}\ \text{\acp{FLOP}/s}} \approx 1.87\ \text{$\mu$s}.
\end{align*}

This extremely low runtime confirms the practical feasibility of real-time implementation on resource-constrained industrial hardware. Furthermore, batching and parallel inference strategies could be adopted to further scale the processing throughput if needed.

\section{Numerical results}\label{sec:numres}

Through numerical simulations we evaluate the detection performance of the proposed \ac{DL}-based method, and compare it to that of alternative existing methods. Then, we assess the performance of the whole system under different traffic conditions, considering various end-to-end metrics.  

\subsection{Simulation setup}

We consider a distributed motion subsystem where controllers are  uniformly distributed within a $6\,\mathrm{m} \times 6\,\mathrm{m} \times 2\,\mathrm{m}$ indoor industrial volume. They are equipped with one transmit antenna and $N_R=4$ receive antennas. 
The $(128, 64)$ binary \ac{LDPC} code standardized in \cite{bluebook} is employed, and all transmissions are \ac{BPSK} modulated. We set the number of decoding iterations to $i_{\max} = 20$, except when explicitely declared. Every coded data block contains $\iota\cdot L_\mathrm{code}=\iota\cdot 128$ symbols, where  $\iota \in [1, 5]$  is chosen uniformly at random for each transmission.  Moreover,  $L_{\text{pre}} = L_{\text{tail}} = 128$. The start and tail sequences are the same for all the controllers. In hexadecimal those sequences are: 
\begin{equation*}
\mathrm{1C71\,B91B\,A9BA\,8457\,B4BC\,5054\,BFD0\,5540}, 
\end{equation*}
\begin{equation*}
\mathrm{AA6C\,CB0C\,C243\,AC5F\,39DC\,7AF4\,640B\,5D95},
\end{equation*} 
respectively\footnote{The tail sequence has been designed using the methods in \cite{giuliani2025access, BattaglioniTAES2025}.}. We simulate controllers transmitting $N_\mathrm{rep} = 2$ replicas within \acp{VF} containing $N_{\mathrm{S}} = 10$ slots.

The wireless channel is modeled as flat Rician fading with log-normal shadowing and distance-dependent path loss, according to \eqref{eq:logno},  
where $\sigma^2_{\mathrm{dB}}=9$, $\lambda_c=0.125\,\mathrm{m}$, $d_{\mathrm{ref}}=1\,\mathrm{m}$, $\beta = 2.1$, $P_t=75$ mW. The \ac{AWGN} power is $\sigma^2_w = -120\,\text{dBm}$. 
We assumed a $K$-factor of $4$ dB. 

The \ac{CNN} has been trained on a dataset of $20 \cdot 10^3$ samples for classes $H_1, \ H_2 \ \text{and} \ H_3$. Instead, since the corresponding occurrences are rare, the $H_4$ class has been trained with $10 \cdot 10^3$ samples. All of those samples are collected from superframe buffers of $N_R$ × $100\,000$ complex symbols (i.e., one sub-buffer  for each antenna). 
The 70$\%$ of these samples were used for training while the remaining 30$\%$ for validation: for the training, the \ac{CNN} made use of the binary cross-entropy loss function and the Adam optimizer (e.g. an adaptive gradient-based optimization algorithm that 	takes advantage of mini-batches to allow efficient training). During simulations, each hyperparameter was individually varied while the others were kept fixed to evaluate its impact on performance \cite{GooBenCou16}. The values for dropout rate, mini-batch size and learning rate that yield the best results are $0.2$, $50$ and $0.001$, respectively.

\subsection{Start and tail sequences detection accuracy}

Let us evaluate the detection accuracy of the proposed multi-label CNN by comparing it with two benchmark approaches: the classic \ac{GLRT} in \cite{Chiani2006}\footnote{In \cite{Chiani2006}, the analysis is carried out for the single-antenna receiver case, where the detection metric is computed directly from the received signal affected by a single realization of noise and interference. In our scenario, instead, each controller is equipped with multiple receive antennas. The signals collected by each antenna experience independent noise realizations and different channel coefficients. To exploit the spatial diversity provided by the antenna array, we combine them using \ac{MRC}, similar to \eqref{eq:mrc_estimation}.} and the \ac{SVM}-based detector proposed in \cite{battaglioni2025mltail} for tail sequence identification. The analysis is structured as follows: first, we compare the ROC performance of the three methods under a fixed traffic level to establish a baseline. Subsequently, we characterize the CNN's detection accuracy across varying traffic intensities $\lambda$ to assess its robustness in different network scenarios.

The classification performance is evaluated using detection rate (recall) $$R  = \mathrm{TP}/(\mathrm{TP} + \mathrm{FN})$$ and false alarm rate $$F = \mathrm{FP}/(\mathrm{FP} + \mathrm{TN}).$$ True positives ($\mathrm{TP}$), true negatives ($\mathrm{TN}$), false positives ($\mathrm{FP}$), and false negatives ($\mathrm{FN}$) refer to correct and incorrect classifications.
Additionally, we consider accuracy, which measures the overall fraction of correct decisions as $$\text{Accuracy} = \frac{\mathrm{TP} + \mathrm{TN}}{\mathrm{TP} + \mathrm{TN} + \mathrm{FP} + \mathrm{FN}}$$and precision, which quantifies the reliability of positive detections as $$\text{Precision} = \frac{\mathrm{TP}}{\mathrm{TP} + \mathrm{FP}}.$$
The multi-label CNN employs sigmoid activations to independently detect the preamble ($A$) and tail ($B$) sequences, with $R$ and $F$ metrics computed for each. The GLRT follows the methodology in \cite{Chiani2006} for both sequences, but the tail detection incorporates a preliminary MRC stage \eqref{eq:mrc_estimation} using preamble-derived channel estimates. Fig.~\ref{fig:ROC} illustrates the ROC performance of the CNN (referred to as CNN$_{01}$), GLRT, and the SVM-based tail detector \cite{battaglioni2025mltail} for a fixed traffic level of $\lambda=0.01$.
As shown in Fig. \ref{fig:ROC}, the CNN-based approach significantly outperforms both the GLRT and the SVM \cite{battaglioni2025mltail}, showing superior robustness under equal false alarm conditions. For example, at a false alarm rate of $0.1$, the CNN achieves detection rates of $0.9658$ for the preamble ($+38.9\%$ over GLRT) and $0.9822$ for the tail ($+10.1\%$ over GLRT and $+28.3\%$ over SVM). 

\begin{figure}[t]
\centering 
    \includegraphics[width=0.4\paperwidth]{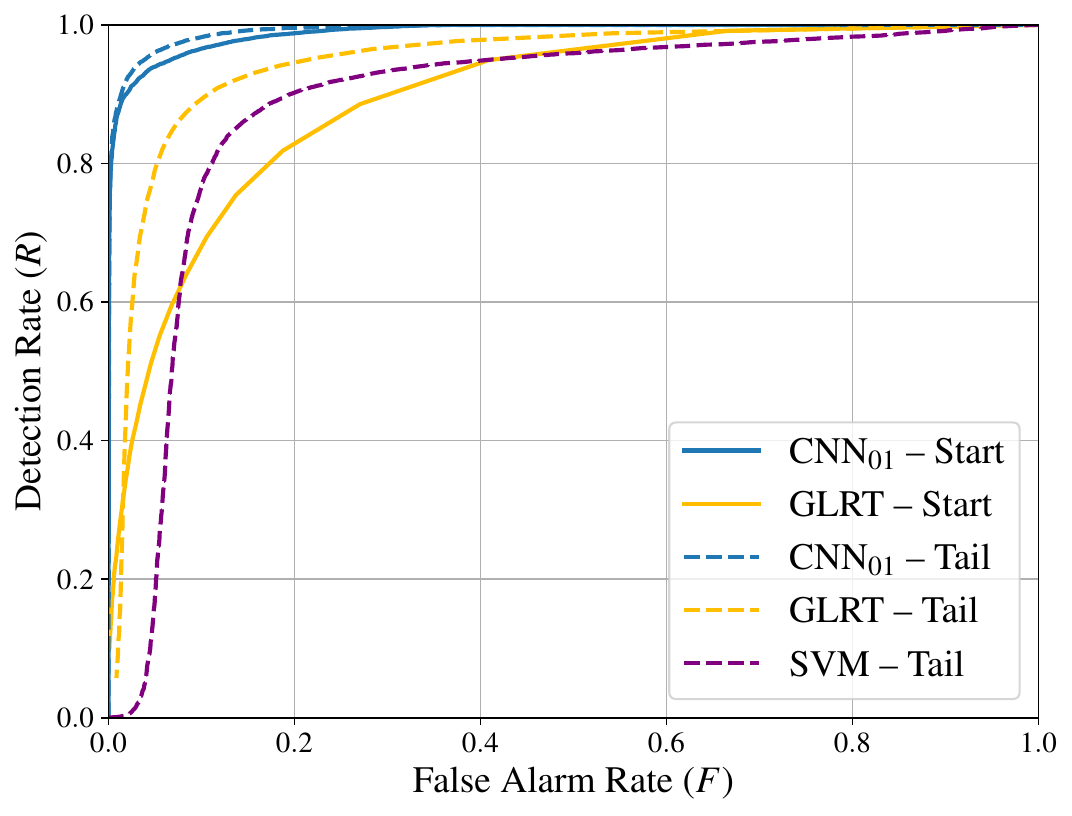} 
    \vspace{-.2cm}\caption{Comparison of \ac{ROC} curves among the \ac{DL}-based detector (CNN$_{01}$), the \ac{GLRT}, and the \ac{SVM} classifier, for a fixed traffic level $\lambda=0.01$.}
    \label{fig:ROC}
\end{figure}

Next, we characterize the CNN-based detector under varying traffic conditions. Three distinct CNN models, denoted as $\text{CNN}_{01}$, $\text{CNN}_{02}$, and $\text{CNN}_{03}$, were trained and evaluated for traffic intensities $\lambda \in \{0.01, 0.02, 0.03\}$, respectively. Solid curves in Fig.~\ref{fig:ROCvarynglambda} illustrate the corresponding ROC curves for both start (\ref{fig:roc_start}) and tail (\ref{fig:roc_tail}) detection across these scenarios. 
Each ROC curve is computed using a test dataset consisting of $10 \cdot 10^3$ samples per class, with the exception of class $\text{H}_4$, for which $5 \cdot 10^3$ samples are considered.
Additionally, Table~\ref{tab:cnn_performance} summarizes the performance metrics (accuracy, precision, and detection rate ($R$)) for each model and label, computed at a fixed false alarm rate ($F$) of $0.1$. As the traffic load $\lambda$ increases, the main performance degradation is observed in recall, while precision remains nearly constant due to the fixed false alarm rate. This indicates that interference primarily leads to missed detections rather than false positives, with only a moderate reduction in accuracy. The impact of missed detections, however, is mitigated in practice, since multiple \ac{SIC} iterations allow replicas not detected in one round to be recovered in subsequent decoding attempts. As expected, the Tail label consistently achieves higher recall than the Start label, showing that the additional information provided by the decoder \acp{LLR} improves detection reliability.

Dashed curves in Fig. \ref{fig:ROCvarynglambda} show the detection performance for the start (\ref{fig:roc_start}) and tail (\ref{fig:roc_tail}) sequences achieved by the CNN$_{01}$ model, (i.e. the one trained with $\lambda=0.01$), when deployed in higher-traffic scenarios. 
It can be observed that in all cases the detection performance achieved by CNN$_{01}$ is only marginally worse than that obtained by the CNNs trained on the corresponding traffic scenario. This demonstrates the robustness and good generalization capability of the designed \ac{NN}, which is able to maintain reliable operation and good performance even when deployed in traffic conditions denser than those considered during training. This is encouraging for real-world deployment, where traffic load is inherently non-stationary and may vary significantly over time.

\begin{figure}[!t] 
    \centering
    \subfloat[Start sequence performance.]{
        \includegraphics[width=0.4\paperwidth]{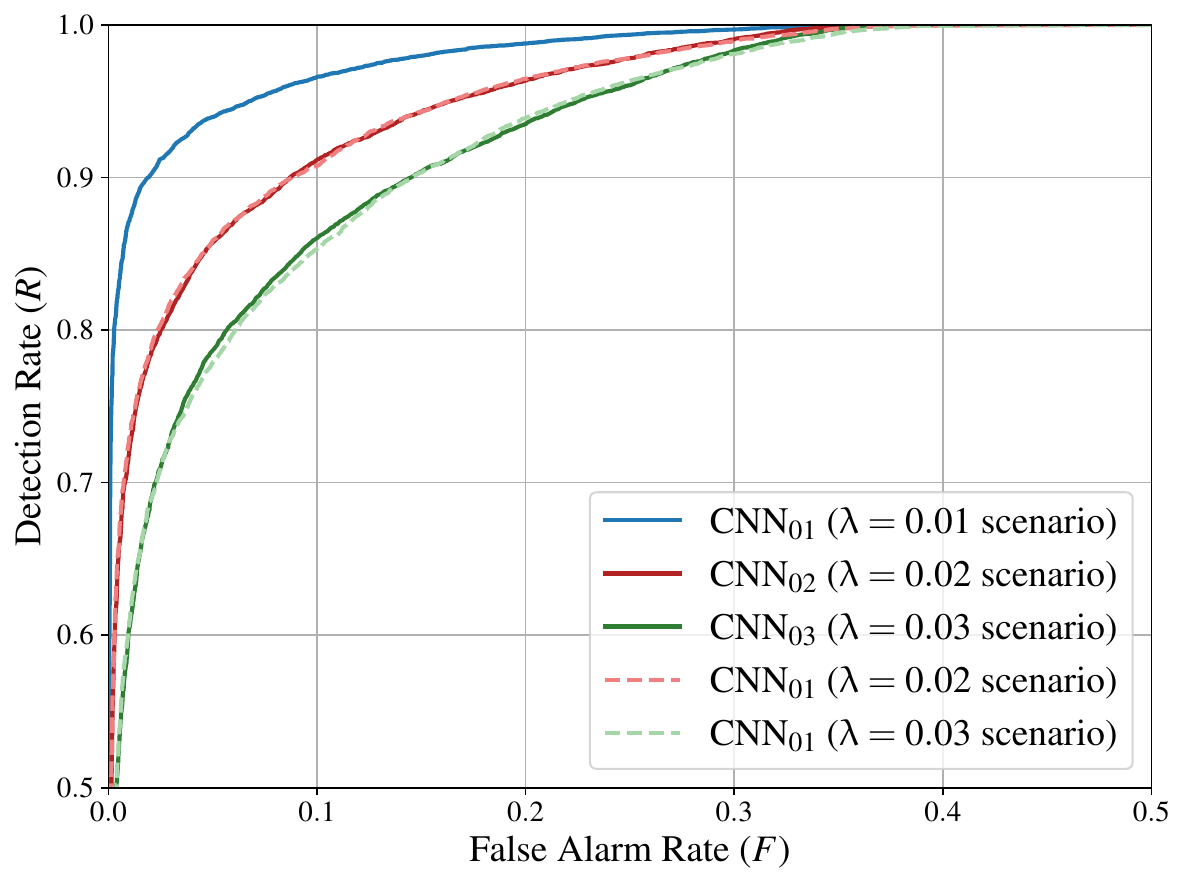}
        \label{fig:roc_start}
    }
    \hfil
    \subfloat[Tail sequence performance.]{
        \includegraphics[width=0.4\paperwidth]{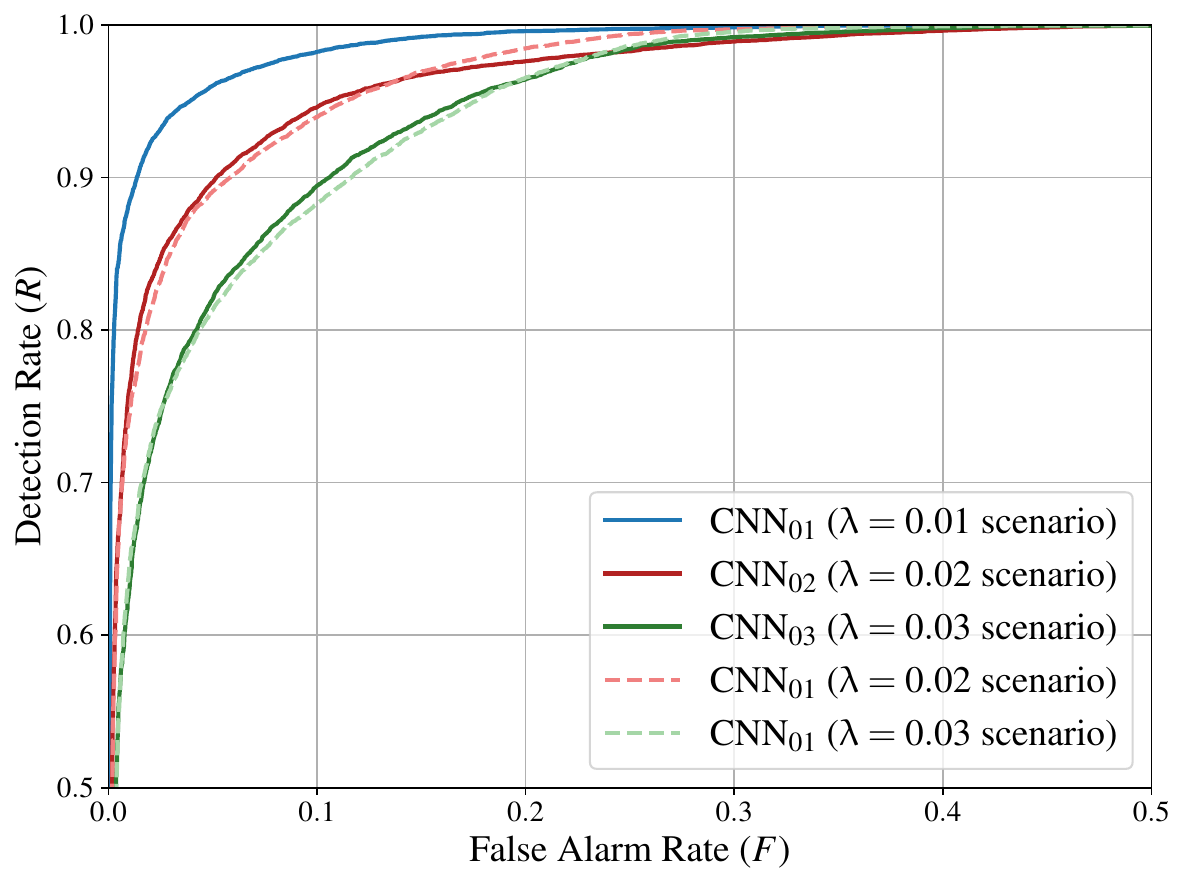}
        \label{fig:roc_tail}
    }
    \caption{ROC curves comparison for different values of $\lambda$. Solid curves represent tests conducted on scenarios with the same $\lambda$ used during training, whereas dashed curves represent mismatched traffic evaluation on CNN$_{01}$.}
    \label{fig:ROCvarynglambda}
\end{figure}
\begin{table}[t]
\centering
\caption{Performance metrics for CNN models trained with different $\lambda$ values, for a fixed false alarm rate $ F=0.1$.}
\label{tab:cnn_performance}
\resizebox{0.9\columnwidth}{!}{%
\renewcommand{\arraystretch}{1.3}
\tiny
\begin{tabular}{c|c|c|c|c}
\hline
\textbf{Model} & \textbf{Label} &  \textbf{Acc.} & \textbf{Prec.} & \textbf{$\bm R$}  \\
\hline
\multirow{2}{*}{CNN$_{\text{\tiny 01}}$} & Start ($A$) & 0.9219 & 0.8285 & 0.9658  \\
& Tail ($B$) & 0.9273 & 0.8307 & 0.9822  \\
\hline
\multirow{2}{*}{CNN$_{\text{\tiny 02}}$} & Start ($A$) & 0.9039 & 0.8203 & 0.9115   \\
& Tail ($B$) & 0.9153 & 0.8257 & 0.9460   \\
\hline
\multirow{2}{*}{CNN$_{\text{\tiny 03}}$} & Start ($A$) & 0.8867 & 0.8116 & 0.8602  \\
& Tail ($B$) & 0.8981 & 0.8175 & 0.8945   \\
\hline
\end{tabular}%
}
\end{table}

\subsection{End-to-end metrics}

We investigate the system performance as a function of the traffic load
$\lambda$ by means of Monte Carlo simulations. In each realization a
single receiver is considered and a finite observation window is defined,
corresponding to one \ac{SF} of $10\,000$ symbols. The \ac{SIC}-aided decoding procedure described in Fig. \ref{fig:receiver_scheme} is applied to the signals received during this interval. 

Let us define the \ac{PLR}  as the ratio between the number of packets that are not recovered at the receiver and the total number of transmitted packets, i.e.,
\[
\text{PLR} \triangleq \frac{N_{\mathrm{err}} + N_{\mathrm{miss}}}{N_{\mathrm{tx}}},
\]
where $N_{\mathrm{err}}$ denotes the number of incorrectly decoded packets and $N_{\mathrm{miss}}$ denotes the number of packets that are not detected at the receiver due to the missed detection of either the start or the tail sequence.  Fig. \ref{fig:plr} reports the \ac{PLR} as a function of the arrival rate $\lambda$ for the \ac{CNN}-based detector and the \ac{GLRT} baseline under different \ac{SF} lengths. For all considered \ac{SF} lengths, the \ac{CNN}-based detector consistently outperforms the \ac{GLRT} approach. A transient effect is also observed when the \ac{SF} length ($10\,000$) is comparable with the \ac{VF} length ($8\,960$). Under this condition, the reduced traffic load at the boundaries of the \ac{SF} leads to a significantly lower \ac{PLR} than that measured for longer \ac{SF} lengths. As the \ac{SF} length increases beyond the \ac{VF} length, this effect diminishes gradually.

Finally, Fig.~\ref{fig:plriter} reports the maximum traffic load $\lambda$ achievable while meeting a target \ac{PLR} of $10^{-3}$, as a function of the  number of decoding iterations $i_{\max}$. The results show that the \ac{CNN}-based detector consistently supports higher loads compared to the \ac{GLRT}-based approach. Overall, the \ac{CNN}-based approach provides a gain between approximately $2.5\times$ and $4.6\times$ in achievable load for the same reliability target.

\begin{figure}
     \centering
     \resizebox{0.4\paperwidth}{!}{%
%
%
\definecolor{mycolor1}{rgb}{0.06600,0.44300,0.74500}%
\definecolor{mycolor2}{rgb}{0.12941,0.12941,0.12941}%
\definecolor{mycolor3}{rgb}{0.66941,0.12941,0.12941}%
\definecolor{mycolor4}{rgb}{0.36941,0.112941,0.72941}%

\begin{tikzpicture}

\begin{axis}[
  width=7.718in,
  height=4.89in,
  at={(1.295in,0.66in)},
  scale only axis,
  xmin=0.00004, xmax=0.0569347972009246,
  xlabel={$\lambda$},
  ymode=log,
  ymin=0.000010719187365197, ymax=1,
  yminorticks=true,
  ylabel={PLR},
  legend style={font=\LARGE, at={(0.5,0.37)}, anchor=north west, legend cell align=left, align=left, draw=white!15!black},
  axis background/.style={fill=white},
  xmajorgrids, ymajorgrids, yminorgrids,
  label style={font=\LARGE, text=mycolor2},
  tick label style={font=\LARGE, text=mycolor2}, 
  tick style={line width=1.1pt},
  major tick length=4.5pt,
  minor tick length=3pt
]

\addplot [color=mycolor1,line width=1.5pt,mark=o,mark size=3.2pt,mark options={solid,mycolor1}]
table[row sep=crcr]{%
0.005 0.000065\\
0.01 0.000243\\
0.015 0.0005\\
0.02 0.00142\\
0.025	0.004818\\
0.03 0.011014\\
0.035 0.040714\\
0.04	0.0678\\
0.045	0.15\\
0.05	0.25\\
};
\addlegendentry{CNN - SF Length $10\,000$}

\addplot [color=mycolor1,line width=1.5pt, dashed , mark=o,mark size=3.2pt,mark options={solid,mycolor1}]
table[row sep=crcr]{%
0.01 0.315327\\
0.0075 0.251383\\
0.005 0.144928\\
0.0025 0.043607\\
0.001 0.003335\\
0.00075 0.00083\\ 
};
\addlegendentry{GLRT - SF Length $10\,000$}

\addplot [color=mycolor2,line width=1.5pt,mark=triangle ,mark size=3.2pt,mark options={solid,mycolor2}]
table[row sep=crcr]{%
0.0015 0.000242
0.002 0.000466\\
0.0025 0.000663\\
0.0035 0.001461\\
0.005 0.009424\\
0.0075 0.13\\
0.01 0.439\\
0.015 0.7359\\
0.02 0.8465\\
};
\addlegendentry{CNN - SF Length $20\,000$}

\addplot [color=mycolor2,dashed,line width=1.5pt,mark=triangle ,mark size=3.2pt,mark options={solid,mycolor2}]
table[row sep=crcr]{%
0.01 0.661954\\
0.0075 0.500471\\
0.005 0.195652\\
0.0025 0.049009\\
0.001 0.003790\\
0.00075 0.00141\\
};
\addlegendentry{GLRT - SF Length $20\,000$}

\addplot[color=mycolor3,line width=1.5pt,mark=x,mark size=3.2pt,mark options={solid,mycolor3}]table[row sep=crcr]{
0.01 0.77\\
0.0075 0.565\\
0.005 0.039\\
0.0025 0.002\\
0.002 0.0013\\
0.001 0.000538\\
};
\addlegendentry{CNN - SF Length $50\,000$}

\addplot[color=mycolor3, dashed, line width=1.5pt,mark=x,mark size=3.2pt, dashed, mark options={solid,mycolor3}]table[row sep=crcr]{
0.01 0.800965\\
0.0075 0.627907\\
0.005 0.287162\\
0.0025 0.047952\\
0.001 0.004425\\
0.00075 0.001691\\
};
\addlegendentry{GLRT - SF Length $50\,000$}


\end{axis}

\end{tikzpicture}%
    }
    \caption{Packet loss rate vs traffic load for different superframe lengths.}
    \label{fig:plr}
\end{figure}
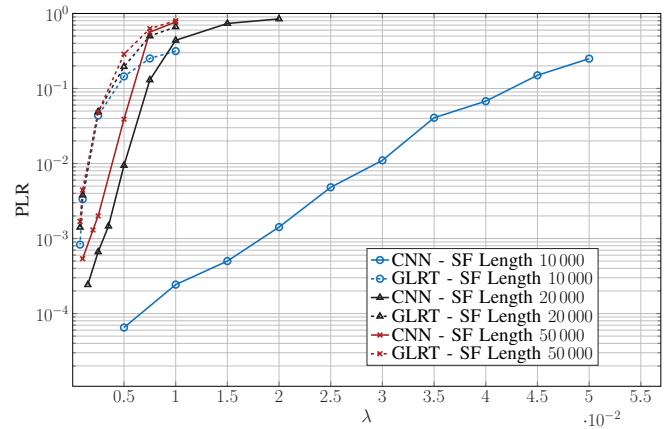

\begin{figure}
     \centering
     \resizebox{0.4\paperwidth}{!}{%
       \begin{tikzpicture}
\begin{axis}[
    width=7.2in,
    height=4.6in,
    xlabel={$i_{\max}$},
    ylabel={$\lambda$ at $\mathrm{PLR}=10^{-3}$},
    xmin=4, xmax=21,
    ymin=0, ymax=0.0034,
    xtick={5,10,15,20},
    ymajorgrids,
    xmajorgrids,
    scaled y ticks=false,
    yticklabel style={
        /pgf/number format/fixed,
        /pgf/number format/precision=4
    },
    legend style={
        font=\large,
        at={(0.03,0.97)},
        anchor=north west,
        draw=white!15!black
    },
    label style={font=\LARGE},
    tick label style={font=\large},
    tick style={line width=1pt},
    line width=1.2pt,
]

\addplot[
    color=blue,
    mark=*,
    mark size=3.2pt,
    line width=1.5pt
]
coordinates {
    (5,  0.00168)
    (10, 0.00215)
    (15, 0.00239)
    (20, 0.00308)
};
\addlegendentry{CNN}

\addplot[
    color=red,
    dashed,
    mark=square*,
    mark size=3.2pt,
    line width=1.5pt
]
coordinates {
    (5,  0.000651)
    (10, 0.000658)
    (15, 0.000667)
    (20, 0.000671)
};
\addlegendentry{GLRT}

\end{axis}
\end{tikzpicture}
    }
\caption{Maximum supported traffic load $\lambda$ at a target  $\mathrm{PLR}=10^{-3}$ as a function of the number of decoding iterations $i_{\max}$.}
    \label{fig:plriter} 
\end{figure}
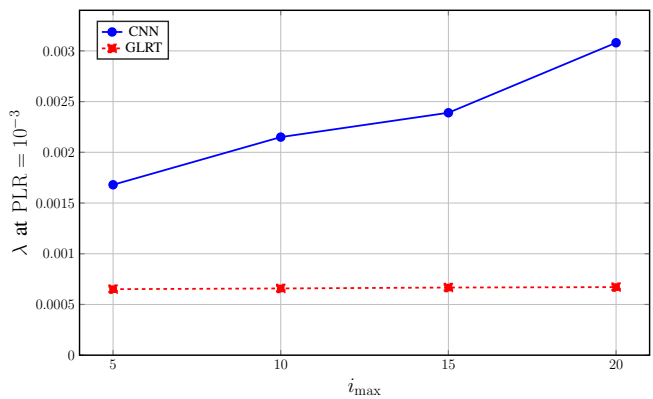

\section{Conclusion}\label{sec:concl}

We investigated asynchronous grant-free \ac{C2C} communications over a shared wireless medium, focusing on the physical-layer challenges arising from uncoordinated transmissions and variable-length packets.
We proposed a receiver architecture in which a single \ac{CNN} performs joint detection of start and tail sequences directly from the received signal. The proposed design integrates learning-based detection with LDPC decoding and \ac{SIC}, enabling full recovery of command units without relying on higher-layer framing or scheduling information.

We also provided a theoretical characterization of the most critical failure event, namely tail confusion caused by asynchronous overlapping transmissions. Under a Poisson activation model, we derived a closed-form expression for the probability of this event, highlighting how the network traffic intensity and replica placement influence the likelihood of erroneous termination decisions at the receiver.

Simulation results demonstrated that the proposed {CNN}-based detector significantly outperforms classical \ac{GLRT}-based methods and recently proposed machine-learning approaches in terms of detection accuracy. When integrated into the full receiver chain, the proposed solution achieves substantially lower \ac{PLR} values across a wide range of traffic conditions. The results also show that the \ac{CNN} maintains robust performance under traffic loads different from those used during training.

\bibliographystyle{IEEEtran}
\bibliography{biblio}

@TECHREPORT{5GACIA_IIoT_2021,
  author      = {{5G-ACIA}},
  title       = {{5G} for {I}ndustrial {I}nternet of {T}hings ({IIoT}): capabilities, features, and potential},
  institution = {{ZVEI -- German Electrical and Electronic Manufacturers' Association}},
  address     = {Frankfurt, Germany},
  year        = {2021},
  month       = oct
}

@BOOK{Blahut2003,
  author    = {Blahut, R. E.},
  title     = {Algebraic Codes for Data Transmission},
  publisher = {Cambridge Univ. Press},
  address   = {Cambridge, U.K.},
  year      = {2003}
}

@ARTICLE{Chang2021,
  author  = {Chang, B. and Li, L. and Zhao, G. and Chen, Z. and Imran, M. A.},
  journal = {IEEE Trans. Commun.},
  title   = {Autonomous {D2D} transmission scheme in {URLLC} for real-time wireless control systems},
  year    = {2021},
  volume  = {69},
  number  = {8},
  pages   = {5546--5558},
  month   = aug,
  doi     = {10.1109/TCOMM.2021.3075680}
}

@ARTICLE{Gao2021,
  author  = {Gao, J. and Zhuang, W. and Li, M. and Shen, X. and Li, X.},
  journal = {IEEE Internet Things J.},
  title   = {{MAC} for machine-type communications in industrial {IoT}---{P}art {I}: {P}rotocol design and analysis},
  year    = {2021},
  volume  = {8},
  number  = {12},
  pages   = {9945--9957},
  month   = jun,
  doi     = {10.1109/JIOT.2021.3051181}
}

@ARTICLE{Aceto2019,
  author  = {Aceto, G. and Persico, V. and Pescap{\'e}, A.},
  journal = {IEEE Commun. Surveys Tuts.},
  title   = {A survey on information and communication technologies for industry 4.0: state-of-the-art, taxonomies, perspectives, and challenges},
  year    = {2019},
  volume  = {21},
  number  = {4},
  pages   = {3467--3501},
  month   = dec,
  doi     = {10.1109/COMST.2019.2938259}
}

@ARTICLE{Cuozzo2025,
  author  = {Cuozzo, G. and Testi, E. and Riolo, S. and Miuccio, L. and Cena, G. and Pasolini, G. and De Nardis, L. and Panno, D. and Chiani, M. and Di Benedetto, M.-G. and Buracchini, E. and Verdone, R.},
  journal = {IEEE Commun. Standards Mag.},
  title   = {Research directions and modeling guidelines for industrial {I}nternet of {T}hings applications},
  year    = {2025},
  volume  = {9},
  number  = {3},
  pages   = {94--103},
  month   = sep,
  doi     = {10.1109/MCOMSTD.2025.3572534}
}

@INPROCEEDINGS{Yuri2017,
  author    = {Polyanskiy, Y.},
  booktitle = {Proc. 2017 IEEE Int. Symp. Inf. Theory},
  title     = {A perspective on massive random-access},
  address={Aachen, Germany},
  month=jun,
  year      = {2017},
  pages     = {2523--2527},
  doi       = {10.1109/ISIT.2017.8006984}
}

@INPROCEEDINGS{FossoNMS,
  author    = {Chen, J. and Fossorier, P. M. C.},
  booktitle = {Proc. 2002 IEEE Global Telecommun. Conf.},
  title     = {Density evolution for {BP}-based decoding algorithms of {LDPC} codes and their quantized versions},
  month=nov,
  address={Taipei, Taiwan},
  year      = {2002},
  pages     = {1378--1382},
  doi       = {10.1109/GLOCOM.2002.1188424}
}

@BOOK{bluebook,
  author    = {{Consultative Committee for Space Data Systems}},
  title     = {{TC} Synchronization and Channel Coding},
  publisher = {CCSDS},
  note      = {{Blue Book CCSDS 231.0-B-4}},
  address   = {Washington, DC, USA},
  month     = jul,
  year      = {2021}
}

@TECHREPORT{3gppTR22804,
  author      = {{3rd Generation Partnership Project (3GPP)}},
  title       = {{TR 22.804 V16.3.0: Study on communication for automation in vertical domains (Release 16)}},
  institution = {3GPP},
  type        = {Technical Report},
  number      = {TR 22.804 V16.3.0},
  year        = {2020},
  month       = jul,
  note        = {Available: \url{https://www.3gpp.org/ftp/Specs/archive/22_series/22.804/22804-g30.zip}}
}

@ARTICLE{Khan2024,
  author  = {Khan, M. U. and Testi, E. and Paolini, E. and Chiani, M.},
  journal = {IEEE Wireless Commun. Lett.},
  title   = {Preamble detection in asynchronous random access using deep learning},
  year    = {2024},
  volume  = {13},
  number  = {2},
  pages   = {279--283},
  month   = feb,
  doi     = {10.1109/LWC.2023.3325918}
}

@INPROCEEDINGS{battaglioni2025mltail,
  author    = {Battaglioni, M. and Giuliani, R. and Chiaraluce, F. and Baldi, M.},
  title     = {Machine learning-based tail sequence detection in {LDPC}-coded space transmissions},
  booktitle = {Proc. 2025 IEEE Wireless Commun. Netw. Conf.},
  address={Milan, Italy},
month=mar,
  year      = {2025},
}

@INPROCEEDINGS{battaglioni_cscn,
  author    = {Battaglioni, M. and Carnevali, E. and De Crescenzo, D. and Testi, E. and Paolini, E.},
  booktitle = {Proc. 2025 IEEE Conf. Standards Commun. Netw.},
  title     = {Boundary detection via deep learning for grant-free asynchronous random access in control-to-control industrial networks},
  year      = {2025},
  month=sep,
  address={Bologna, Italy},
  doi       = {10.1109/CSCN67557.2025.11230747}
}

@ARTICLE{Decre2025,
  author  = {De Crescenzo, D. and Testi, E. and Paolini, E.},
  journal = {IEEE Wireless Commun. Lett.},
  title   = {{Deep learning for replica detection and combining in asynchronous grant-free mMTC}},
  year    = {2025},
  volume  = {15},
  pages   = {225--229},
  month   = oct,
  doi     = {10.1109/LWC.2025.3621311}
}

@ARTICLE{BattaglioniTAES2025,
  author  = {Battaglioni, M. and Andrews, K. and Giuliani, R. and Marinelli, F. and Chiaraluce, F. and Baldi, M.},
  journal = {IEEE Trans. Aerosp. Electron. Syst.},
  title   = {{Design and analysis of the tail sequence for short LDPC-coded space communications}},
  year    = {2025},
  volume  = {61},
  number  = {6},
  pages   = {16938--16952},
  month   = dec,
  doi     = {10.1109/TAES.2025.3599498}
}

@ARTICLE{Paolini2015,
  author  = {Paolini, E. and Liva, G. and Chiani, M.},
  journal = {IEEE Trans. Inf. Theory},
  title   = {{Coded slotted ALOHA: a graph-based method for uncoordinated multiple access}},
  year    = {2015},
  volume  = {61},
  number  = {12},
  pages   = {6815--6832},
  month   = dec,
  doi     = {10.1109/TIT.2015.2492579}
}

@ARTICLE{DeGaud2014,
  author  = {De Gaudenzi, R. and Del R{\'i}o Herrero, O. and Acar, G. and Garrido Barrab{\'e}s, E.},
  journal = {IEEE Trans. Wireless Commun.},
  title   = {{Asynchronous contention resolution diversity ALOHA: making CRDSA truly asynchronous}},
  year    = {2014},
  volume  = {13},
  number  = {11},
  pages   = {6193--6206},
  month   = nov,
  doi     = {10.1109/TWC.2014.2334620}
}

@ARTICLE{Casini2007,
  author  = {Casini, E. and De Gaudenzi, R. and Del R{\'i}o Herrero, O.},
  journal = {IEEE Trans. Wireless Commun.},
  title   = {{Contention resolution diversity slotted ALOHA (CRDSA): an enhanced random access scheme for satellite access packet networks}},
  year    = {2007},
  volume  = {6},
  number  = {4},
  pages   = {1408--1419},
  month   = apr,
  doi     = {10.1109/TWC.2007.348337}
}

@BOOK{GooBenCou16,
  author    = {Goodfellow, I. and Bengio, Y. and Courville, A.},
  title     = {Deep Learning},
  publisher = {MIT Press},
  year      = {2016},
  note      = {\url{http://www.deeplearningbook.org}}
}

@INPROCEEDINGS{Azari2017,
  author    = {Azari, A. and Popovski, P. and Miao, G. and Stefanovic, C.},
  booktitle = {Proc. 2017 IEEE Global Commun. Conf.},
  title     = {{Grant-free radio access for short-packet communications over 5G networks}},
 address = {Singapore},
 month=dec,
 year=2017,
  doi       = {10.1109/GLOCOM.2017.8255054}
}

@ARTICLE{deSou2023,
  author  = {de Souza, J. H. I. and Abr{\~a}o, T.},
  journal = {IEEE Syst. J.},
  title   = {Deep learning-based activity detection for grant-free random access},
  year    = {2023},
  volume  = {17},
  number  = {1},
  pages   = {940--951},
  month   = mar,
  doi     = {10.1109/JSYST.2022.3175658}
}

@ARTICLE{giuliani2025access,
  author  = {Giuliani, R. and Battaglioni, M. and Baldi, M. and Chiaraluce, F. and Maturo, N.},
  journal = {IEEE Access},
  title   = {{Telecommand rejection probability in CCSDS-compliant LDPC-coded space transmissions with tail sequence}},
  year    = {2025},
  volume  = {13},
  pages   = {8924--8940},
  month=jan,
  doi     = {10.1109/ACCESS.2025.3527057}
}

@ARTICLE{Chiani2006,
  author  = {Chiani, M. and Martini, M. G.},
  journal = {IEEE Trans. Commun.},
  title   = {{On sequential frame synchronization in AWGN channels}},
  year    = {2006},
  volume  = {54},
  number  = {2},
  pages   = {339--348},
  month   = feb,
  doi     = {10.1109/TCOMM.2005.863727}
}

@ARTICLE{Hu2024,
  author  = {Hu, Y. and Jin, H. and Seo, J.-B.},
  journal = {IEEE Trans. Veh. Technol.},
  title   = {Asynchronous Random Access Systems With Immediate Collision Resolution for Low Power Wide Area Networks},
  year    = {2024},
  volume  = {73},
  number  = {2},
  pages   = {2755--2770},
  month   = feb,
  doi     = {10.1109/TVT.2023.3320353}
}

@INPROCEEDINGS{Almonacid2017,
  author    = {Almonacid, V. and Franck, L.},
  booktitle = {Proc. 2017 IEEE Int. Conf. Commun.},
  title     = {An asynchronous high-throughput random access protocol for low power wide area networks},
  month=may,
  year      = {2017},
  address= {Paris, France},
  doi       = {10.1109/ICC.2017.7996382}
}

@ARTICLE{Jiang2022,
  author  = {Jiang, H. and Qu, D. and Ding, J. and Wang, Z. and He, H. and Chen, H.},
  journal = {IEEE Wireless Commun.},
  title   = {Enabling {LPWAN} Massive Access: Grant-Free Random Access with Massive {MIMO}},
  year    = {2022},
  volume  = {29},
  number  = {4},
  pages   = {72--77},
  month   = aug,
  doi     = {10.1109/MWC.102.2100276}
}

@ARTICLE{Azari2021,
  author  = {Azari, A. and Stefanovi{\'c}, {\v{C}}. and Popovski, P. and Cavdar, C.},
  journal = {IEEE Trans. Green Commun. Netw.},
  title   = {Energy-Efficient and Reliable {IoT} Access Without Radio Resource Reservation},
  year    = {2021},
  volume  = {5},
  number  = {2},
  pages   = {908--920},
  month   = jun,
  doi     = {10.1109/TGCN.2021.3051033}
}

@INPROCEEDINGS{Ebrahimi2017,
  author    = {Ebrahimi, M. and Lahouti, F. and Kostina, V.},
  booktitle = {Proc. 2017 IEEE Int. Symp. Inf. Theory},
  title     = {Coded random access design for constrained outage},
  month=jun,
  year      = {2017},
  address = {Aachen, Germany},
  pages     = {2732--2736},
  doi       = {10.1109/ISIT.2017.8007026}
}

@ARTICLE{Zhu2021,
  author  = {Zhu, W. and Tao, M. and Yuan, X. and Guan, Y.},
  journal = {IEEE Trans. Wireless Commun.},
  title   = {Deep-Learned Approximate Message Passing for Asynchronous Massive Connectivity},
  year    = {2021},
  volume  = {20},
  number  = {8},
  pages   = {5434--5448},
  month   = aug,
  doi     = {10.1109/TWC.2021.3067903}
}

@ARTICLE{Liva2024,
  author  = {Liva, G. and Polyanskiy, Y.},
  journal = {Proc. IEEE},
  title   = {Unsourced Multiple Access: A Coding Paradigm for Massive Random Access},
  year    = {2024},
  volume  = {112},
  number  = {9},
  pages   = {1214--1229},
  month   = sep,
  doi     = {10.1109/JPROC.2024.3437208}
}

@ARTICLE{Choi2022,
  author  = {Choi, J. and Ding, J. and Le, N.-P. and Ding, Z.},
  journal = {IEEE Wireless Commun.},
  title   = {Grant-Free Random Access in Machine-Type Communication: Approaches and Challenges},
  year    = {2022},
  volume  = {29},
  number  = {1},
  pages   = {151--158},
  month   = feb,
  doi     = {10.1109/MWC.121.2100135}
}

\end{document}